\documentclass[a4paper]{article}

\usepackage{paralist}
\usepackage{todonotes}
\usepackage{tikz} 
\usetikzlibrary{arrows}
\usetikzlibrary{arrows,automata}
\usetikzlibrary{automata,positioning,arrows}
\usetikzlibrary{shapes}
\usetikzlibrary{shapes,backgrounds}
\usetikzlibrary{matrix}
\usetikzlibrary{chains,fit,shapes}
\usepackage{amsmath}
\usepackage{amssymb}
\usepackage{amsfonts}
\usepackage{amsthm}
\usepackage[ruled]{algorithm2e}

\usepackage{wrapfig}

\usepackage[top=1in, bottom=1.25in, left=1.25in, right=1.25in]{geometry}

\usepackage{url}
\usepackage[hidelinks]{hyperref}                
\usepackage{authblk}

\newtheorem{theorem}{Theorem}
\newtheorem{lemma}[theorem]{Lemma}
\newtheorem{proposition}[theorem]{Proposition}
\newtheorem{definition}[theorem]{Definition}
\newtheorem{corollary}[theorem]{Corollary}

\theoremstyle{remark}
\newtheorem{example}[theorem]{Example}
\newtheorem{remark}[theorem]{Remark}
\newtheorem{observation}[theorem]{Observation}

\DeclareMathOperator{\npclass}{\mathsf{NP}}
\DeclareMathOperator{\wclass}{\mathsf{W}}
\DeclareMathOperator{\conpclass}{\mathsf{coNP}}
\DeclareMathOperator{\bigO}{O}
\DeclareMathOperator{\regex}{\mathsf{RX}}

\DeclareMathOperator{\regexp}{\mathsf{RE}}
\DeclareMathOperator{\MFA}{\mathsf{MFA}}
\DeclareMathOperator{\NFA}{\mathsf{NFA}}
\DeclareMathOperator{\DFA}{\mathsf{DFA}}
\DeclareMathOperator{\MDMFA}{\mathsf{md-MFA}}

\newcommand{\kSigma}{\Sigma_k}

\newcommand{\eSigma}{\Sigma_{\eword}}
\newcommand{\ekSigma}{\Sigma_{\eword, k}}
\newcommand{\emSigma}{\Sigma_{\eword, m}}
\DeclareMathOperator{\samereach}{\hspace{-0.15em}\Downarrow\hspace{-0.1em}}
\DeclareMathOperator{\notsamereach}{\hspace{-0.15em}\not\Downarrow\hspace{-0.1em}}
\DeclareMathOperator{\memSync}{\Bumpeq}
\DeclareMathOperator{\notmemSync}{\not\Bumpeq}

\DeclareMathOperator{\mdet}{\textsf{md}}

\DeclareMathOperator{\minVarPara}{\mathsf{minvar}}

\DeclareMathOperator{\contracted}{\mathsf{con}}
\DeclareMathOperator{\deltaContr}{\delta_{\contracted}}

\DeclareMathOperator{\paraAlphabet}{\Gamma}
\DeclareMathOperator{\eword}{\varepsilon}
\DeclareMathOperator{\lang}{\mathcal{L}}
\DeclareMathOperator{\reflang}{\mathfrak{R}}

\DeclareMathOperator{\var}{\textsf{var}}
\DeclareMathOperator{\altop}{\vee}
\DeclareMathOperator{\lpara}{\Rsh}
\DeclareMathOperator{\rpara}{\Lsh}
\DeclareMathOperator{\refexp}{\textsf{ref}}
\newcommand{\deref}{\mathcal{D}}

\DeclareMathOperator{\opened}{\mathtt{O}}
\DeclareMathOperator{\closed}{\mathtt{C}}

\DeclareMathOperator{\oneinthreethreesat}{\textsc{1-in-3 3SAT}}
\DeclareMathOperator{\threesat}{\textsc{3SAT}}
\DeclareMathOperator{\nondef}{\bot}
\DeclareMathOperator{\synchmeminstrTable}{\mathbf{T}}

\DeclareMathOperator{\lce}{\textsf{LCE}}
\DeclareMathOperator{\prefrel}{\preceq_{\textsf{pref}}}
\DeclareMathOperator{\SynchMemAlgo}{\textsc{SyncMembership}}

\DeclareMathOperator{\nonSyncBranch}{\mathsf{NSB}}

\DeclareMathOperator{\memInstAlphabet}{\Gamma}

\DeclareMathOperator{\canonicalNFA}{\mathcal{R}}
\DeclareMathOperator{\canonicalMFA}{\mathcal{M}}

\DeclareMathOperator{\preRel}{\vartriangleright_{\textsf{def}}}
\DeclareMathOperator{\postRel}{\vartriangleright_{\textsf{call}}}
\DeclareMathOperator{\notpostRel}{\not\vartriangleright_{\textsf{call}}}
\DeclareMathOperator{\notpreRel}{\not\vartriangleright_{\textsf{def}}}

\DeclareMathOperator{\subexpression}{\mathsf{se}}
\DeclareMathOperator{\syntaxtree}{\mathcal{T}}
\DeclareMathOperator{\nodes}{\mathcal{N}}

\DeclareMathOperator{\nodelabel}{\mathsf{type}}

\DeclareMathOperator{\altoptree}{[\altop]}
\DeclareMathOperator{\concoptree}{[\cdot]}

\DeclareMathOperator{\plusoptree}{[+]}

\DeclareMathOperator{\memorylist}{\mathfrak{M}}
\DeclareMathOperator{\notinuse}{\bot}

\DeclareMathOperator{\instate}{\mathsf{in}}
\DeclareMathOperator{\outstate}{\mathsf{out}}
\DeclareMathOperator{\interstate}{\mathsf{m}}

\DeclareMathOperator{\actvarset}{\mathsf{avs}}

\DeclareMathOperator{\avd}{\mathsf{avd}}

\DeclareMathOperator{\savd}{\mathsf{savd}}

\DeclareMathOperator{\crudeAutomaton}{\mathcal{H}}
\DeclareMathOperator{\crudeAutomatonNodes}{\mathcal{N}_{\crudeAutomaton}}

\DeclareMathOperator{\compressedInstructions}{\Phi}

\newcommand{\open}[1]{\mathtt{o}({#1})}
\newcommand{\close}[1]{\mathtt{c}({#1})}
\newcommand{\openshort}[1]{\mathtt{o}{#1}}
\newcommand{\closeshort}[1]{\mathtt{c}{#1}}

\newcommand{\ta}{\ensuremath{\mathtt{a}}}
\newcommand{\tb}{\ensuremath{\mathtt{b}}}
\newcommand{\tc}{\ensuremath{\mathtt{c}}}
\newcommand{\td}{\ensuremath{\mathtt{d}}}

\newcommand{\tzero}{\ensuremath{\mathtt{0}}}

\title{Regular Expressions with Backreferences:\\ Polynomial-Time Matching Techniques}

\author[1]{Markus L. Schmid}

\affil[1]{Humboldt-Universit\"at zu Berlin, Berlin, Germany, \texttt{MLSchmid@MLSchmid.de}}

\date{\vspace{-1cm}}

\begin{document}

\maketitle

\begin{abstract}
Regular expressions with backreferences (regex, for short), as supported by most modern libraries for regular expression matching, have an NP-complete matching problem. We define a complexity parameter of regex, called active variable degree, such that regex with this parameter bounded by a constant can be matched in polynomial-time. Moreover, we formulate a novel type of determinism for regex (on an automaton-theoretic level), which yields the class of memory-deterministic regex that can be matched in time $O(|w|p(|\alpha|))$ for a polynomial $p$ (where $\alpha$ is the regex and $w$ the word). Natural extensions of these concepts lead to  properties of regex that are intractable to check.
\end{abstract}

\section{Introduction}\label{sec:intro}

Regular expressions were first introduced by Kleene in 1956~\cite{Kleene1956} as a theoretical concept (an early implementation is due to Thompson~\cite{Thompson1968}). Since then, they have been enriched with practically motivated extensions and modifications, which is mainly due to their rather high practical relevance (see the IEEE POSIX standard~\cite{ieee:posix} or the W3C recommendations~\cite{w3c:dtd, w3c:xsd, w3c:sparql11}, see \cite{Friedl2006} for an overview of the role of regular expressions as a practical tool, and also note that variants of regular expressions are intensively investigated in the database theory community, due to their relevance for graph databases (see, e.\,g., \cite{LibkinVrgoc2012, BarceloEtAl2012, BjorklundEtAl2015, LosemannMartens2013, BaganEtAl2013, BarceloEtAl2013_2, MartensTrautner2018} or \cite{Barcelo2013} for a survey) and the information extraction framework of document spanners (see, e.\,g., \cite{AmarilliEtAl2021, FlorenzanoEtAl2020, FaginEtAl2015, FreydenbergerHolldack2018, SchmidSchweikardt2022b, SchmidSchweikardt2021a, SchmidSchweikardt2021b} or \cite{SchmidSchweikardt2022a} for a survey). Regular expressions have excellent decidability- and complexity-properties, while at the same time providing expressive power that is sufficient for many important computational tasks. Most of the practical enhancements added over the years are mere ``syntactic sugar'' and therefore preserve these positive properties. However, adding so-called \emph{backreferences} drastically increases expressive power and therefore leads to intractability and even undecidability.\par
A \emph{backreference} in a regular expression is a possibility to repeat the subword matched to a specific subexpression. For example, the $x\{\dots\}$-construct in the expression $r = x\{(\ta \altop \tb)^*\} \tc x$ stores in variable $x$ whatever subword is matched to the subexpression $(\ta \altop \tb)^*$, and the following occurrence of variable $x$ then refers to exactly this subword (thus, $r$ describes  the non-regular language $\{w \tc w \mid w \in \{\ta, \tb\}^*\}$). In the following, we denote regular expressions with such backreferences by the term \emph{regex}. The \emph{matching problem} of regex, i.\,e., deciding whether a given regex can match a given word, is $\npclass$-complete (even for strongly restricted variants) \cite{Aho1990, FernauSchmid2015, FernauEtAl2016, FernauEtAl2020}, and decision problems like inclusion, equivalence and universality are undecidable~\cite{Freydenberger2013} (even if the input expressions only use one variable with only a bounded number of occurrences). Nevertheless, regular expression libraries of almost all modern programming languages (like, e.\,g., Java, PERL, Python and .NET) support backreferences (although they syntactically and even semantically slightly differ from each other (see the discussion in~\cite{FreydenbergerSchmidJCSS})), and they are even part of the POSIX standard~\cite{ieee:posix}.

\subsection{The Regex Matching Problem}

The arguably most important problem for practical considerations is the \emph{matching problem}. Its general $\npclass$-completeness was shown in~\cite{Aho1990}, but also follows from matching \emph{patterns with variables}~\cite{Angluin1980}, i.\,e., checking whether the variables $x_i$ in a pattern $\alpha \in (\Sigma \cup \{x_i \mid i \in \mathbb{N}\})^*$ can be uniformly replaced by words from $\Sigma^*$ in order to obtain a given word (see~\cite{ManeaSchmid2019} for a survey or the more recent publications~\cite{FernauEtAl2020, DayEtAl2017, DayEtAl2018}). These patterns are a quite successful tool for obtaining negative results for regex,\footnote{The undecidability results of~\cite{Freydenberger2013} also follow from the fact that regex can describe systems of patterns.} but the many known positive algorithmic approaches to matching patterns (see~\cite{ReidenbachSchmid2014, FernauEtAl2015, DayEtAl2017, DayEtAl2018}) are tailored to the ``backreferencing-aspect'' and seem unfit for handling the ``regular expression-aspect'' of regex. In fact, even though there are many deep theoretical (yet negative) results about the complexity and decidability of regex, positive algorithmic approaches are rather scarce. \par
In~\cite{FreydenbergerSchmidJCSS}, \emph{deterministic} regex (\emph{det-regex}) are introduced.\footnote{Deterministic (classical) regular expressions are an established concept~\cite{CzerwinskiEtAl2017, Gelade2012, LosemannEtAl2016, GrozManeth2017}.} Since they are characterised via a purely deterministic automaton model, they can be matched efficiently and, if further restricted, they have some decidable problems in static analysis (their language theoretical properties have been thoroughly investigated in~\cite{FreydenbergerSchmidJCSS}). However, if efficient matchability is our main concern, det-regex seem unnecessarily restricted, since they do not cover all regular languages. In fact, det-regex cover very well what it means for a regex to be deterministic in the strongest possible way, but not quite what it means to be ``easily matchable''. 

\subsection{Our Contribution}

We develop two different approaches to efficient regex matching:

\begin{itemize}
\item \textbf{Regex with bounded active variable degree}:
We define a complexity parameter of regex, called \emph{active variable degree} (denoted by $\avd(\alpha)$), and show that regex can be matched in time $|\alpha||w|^{\bigO(\avd(\alpha))}$. Note that $|\alpha||w|^{\bigO(\var(\alpha))}$is a trivial upper bound, where $\var(\alpha)$ is the total number of variables of $\alpha$, and that $\avd(\alpha)$ is always upper bounded by $\var(\alpha)$. Intuitively speaking, the parameter $\avd(\alpha)$ measures the number of variables that can be \emph{active} at the same time in a match, and the algorithmic application relies in devising a matching procedure, which, in a sense, reuses variables that are currently \emph{not active}. This approach can also be seen as a technique to reduce the number of variables of a regex, a problem that, in its general form, is undecidable (see~\cite{Freydenberger2013}).
\item \textbf{Memory-deterministic regex}: We come up with a possibility to limit the inherent non-determinism of regex to those parts that have nothing to do with backreferences. The thus obtained class of \emph{memory-deterministic regex} enforces some synchronisation between different computational branches in a matching procedure, and therefore can be matched in time $p(|\alpha|)|w|$ for a (low-degree) polynomial $p$. This means that matching memory-deterministic regex can be done in time linear in $|w|$ if measured in \emph{data complexity}.\footnote{Data complexity is motivated by considerations in database theory; in this regard, a regex can be seen as a (usually short) query that is to be evaluated on a (potentially large) data-object, i.\,e., the word.} This is worth pointing out, since the full class of regex can most likely not be matched in time $f(|\alpha|)g(|w|)$ for \emph{any} polynomial $g$ and \emph{computable} function $f$, or for \emph{any} polynomial $f$ and \emph{computable} function $g$ (this follows from the $\wclass[1]$-hardness of the problem if parameterised by the size of the regex or by the size of the input word~\cite{FernauEtAl2016}). The concept of memory determinism is rather complicated, since it cannot be achieved by some local and syntactical restrictions. Hence, a main challenge is to show that memory determinism can be checked efficiently.
\end{itemize}

These positive results are complemented with lower bounds. The active variable degree can be improved to a much stronger complexity parameter (that also can be exploited in similar ways), but computing it is $\conpclass$-hard. The development of memory determinism is carefully governed by intractability results as follows: First, we show that even rather strong restrictions of non-determinism will lead to an intractable matching problem, as long as these restrictions are of a local and syntactical nature. This observation leads to a regex-property that is entirely non-syntactic in the sense that it is formulated with respect to the possible matchings. While this property is sufficient for efficient matching, it is also $\conpclass$-hard to be checked for. The concept of memory determinism results from finding a balance between matching-complexity and the complexity of checking the property. 

\subsection{Techniques}

Our main algorithmic tool is \emph{memory automata} ($\MFA$), a recently introduced automaton-based characterisation of regex (see~\cite{Schmid2016, FreydenbergerSchmidJCSS}). If regex are represented as $\MFA$, their structure is much easier to analyse and we can conveniently abstract from the actual backreferences by interpreting an $\MFA$ as an $\NFA$ that accepts a regular language with special meta-symbols. This point of view is vital and provides the necessary leverage for developing our concepts and proving the respective results. In this way, we are able to define and exploit the active variable degree by analysing the automaton-structure underlying the regex, and the development of memory determinism will also be done on the level of $\MFA$. To the knowledge of the author, restricting non-determinism by talking about computations of the automaton rather than syntactical properties is a novel approach.

\section{Preliminaries}\label{sec:prelim}

Let $\mathbb{N} = \{1, 2, 3, \ldots\}$ and $[n] = \{1, 2, \ldots, n\}$, $n \in \mathbb{N}$.  For a set $A$, by $\mathcal{P}(A)$ we denote its power set. For a string $w$, $|w|$ denotes its length and, for every $i \in [|w|]$, $w[i]$ denotes the $i^{\text{th}}$ symbol of $w$. Moreover, by $w[i..j]$, we denote the factor of $w$ from symbol $i$ to symbol $j$, and for $b \in \Sigma$, $|w|_{b}$ denotes the number of occurrences of $b$ in $w$. The symbol $\eword$ denotes the empty word. For an alphabet $A$, $A^+$ denotes the set of non-empty words over $A$ and $A^* = A^+ \cup \{\eword\}$; we set $A_{\eword} = A \cup \{\eword\}$ (i.\,e., we also use $\eword$ as a symbol denoting the empty word). For any language descriptor $D$, $\lang(D)$ denotes the language of $D$.

\subsection{Regular Expressions with Backreferences} 

Let $X$ denote a finite set of \emph{variables}. The set $\regex_{\Sigma, X}$ of \emph{regular expressions with backreferences} (\emph{over $\Sigma$ and $X$}), also denoted by \emph{regex}, for short, is recursively defined as follows:
\begin{compactenum}
\item\label{regexDefPointOne} $a \in \regex_{\Sigma, X}$ and $\var(a) = \emptyset$, for every $a \in \eSigma$,
\item\label{regexDefPointTwo} $(\alpha \cdot \beta) \in \regex_{\Sigma, X}$, $(\alpha \altop \beta) \in \regex_{\Sigma, X}$, and $(\alpha)^+ \in \regex_{\Sigma, X}$, for every $\alpha,\beta\in \regex_{\Sigma, X}$;\\ furthermore, $\var((\alpha \cdot \beta)) = \var((\alpha \altop \beta)) = \var(\alpha) \cup \var(\beta)$ and $\var((\alpha)^+) = \var(\alpha)$, 
\item $x \in \regex_{\Sigma, X}$ and $\var(x) = \{x\}$, for every $x \in X$, 
\item $x\{\alpha\} \in \regex_{\Sigma, X}$ and $\var(x\{\alpha\}) = \var(\alpha) \cup \{x\}$, for every $\alpha \in \regex_{\Sigma, X}$ and $x \in X \setminus \var(\alpha)$.
\end{compactenum}
For $\alpha \in \regex_{\Sigma, X}$, we set $\alpha^* = \alpha^+ \altop \eword$, and we usually omit the operator `$\cdot$'. In a regex, we call an occurrence of symbol $x \in X$ a \emph{recall of variable $x$} and a subexpression of the form $x\{\alpha\}$ a \emph{definition of variable $x$}; if we just talk about (\emph{occurrences of}) \emph{variables}, then we refer to a recall or a definition. The subset of $\regex_{\Sigma, X}$ that can be created by Points~\ref{regexDefPointOne}~and~\ref{regexDefPointTwo} is exactly the set of \emph{regular expressions} over $\Sigma$, which we also call \emph{classical} regular expressions. \par
The \emph{syntax tree} $\syntaxtree(\alpha)$ of $\alpha \in \regex_{\Sigma, X}$ with $X = [m]$ with nodes $\nodes(\alpha)$ is defined as follows. 
\begin{itemize}
\item If $\alpha \in \Sigma_{\eword} \cup X$, then $\syntaxtree(\alpha)$ is a single node labelled with $[\alpha]$.
\item If $\alpha = (\beta \altop \gamma)$ (or $\alpha = (\beta \cdot \gamma)$), then the root of $\syntaxtree(\alpha)$ is labelled with $\altoptree$ (or $\concoptree$, respectively) and has the root of $\syntaxtree(\beta)$ as its left and the root of $\syntaxtree(\gamma)$ as its right child.
\item If $\alpha = (\beta)^+$ (or $\alpha = (x\{\beta\})$), then the root of $\syntaxtree(\alpha)$ is labelled with $\plusoptree$ (or $[x\{\}]$, respectively) and has the root of $\syntaxtree(\beta)$ as its only child.
\end{itemize}

\begin{figure}
\centering
\scalebox{0.6}{\includegraphics{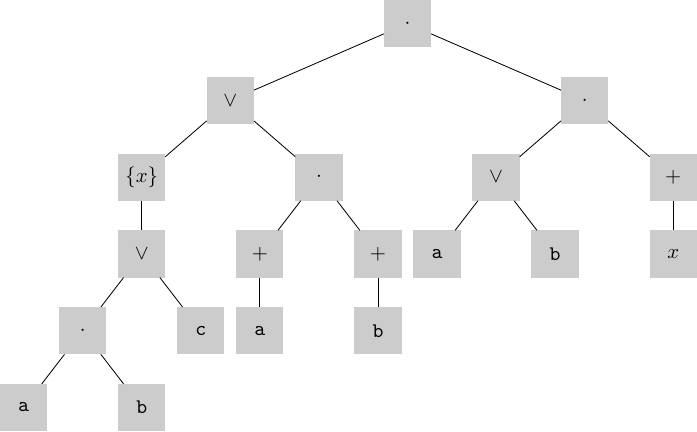}}
\caption{The syntax tree $\syntaxtree(\alpha)$ for $\alpha = (x\{\ta \tb \altop \tc\} \altop (\ta^+ \tb^+)) \: (\ta \altop \tb) \: x^+$.}
\label{Fig:STExample}
\end{figure}

See Fig.~\ref{Fig:STExample} for an illustration of a syntax tree. By $\subexpression(t)$, we refer to the subexpression of $\alpha$ that corresponds to a node $t \in \nodes(\alpha)$, i.\,e., $\subexpression(t)$ is the regex represented by the subtree of $\syntaxtree(\alpha)$ rooted by $t$. By $\nodelabel(t)$, we refer to the label of node $t \in \nodes(\alpha)$. In the following, we assume that regex are always given as syntax trees $\syntaxtree(\alpha)$; in particular, we set $|\alpha| = |\syntaxtree(\alpha)|$. \par
On an intuitive level, the semantics of a regex are clear: the expression is matched to a word as for classical regular expressions and if we encounter a definition $x\{\beta\}$, then the factor $v$ matched to $\beta$ is stored in $x$ and further occurrences of $x$ are treated as $v$ (as a particularity, \emph{undefined} variables are treated as $\eword$). However, several particularities, that are easily overlooked on this intuitive level, complicate the behaviour of regex considerably. For example, in $(x\{\ta^*\} \altop (x\{\tb^*\} x)) (x \altop x)$, depending on the alternations, either the first or the second definition of $x$ is instantiated and either the second or third recall of $x$ is instantiated, while the first recall is instantiated if and only if the second definition is. Moreover, the second and third recall of $x$ can refer to the first or the second definition of $x$, while the first recall can only refer to the second definition. As a result, whether or not the second or third recall of $x$ refers to the same factor as the first recall depends on the alternations. The situation is even more complicated  by operator $+$ as, e.\,g., in $\alpha = ((x\{\ta^*\} \altop (y\{\tb^*\} y)) (x \altop y))^+$. Now there is a potentially unbounded number of instances of each of the definitions and recalls of variables $x$ and $y$, and the allocation between definitions and recalls can reach over several iterations of the operator $+$. For example, if the definition of $x$ is instantiated in the first, and the definition of $y$ in the second to fifth iteration, then the recall of $x$, if instantiated in the fifth iteration, refers to the definition of $x$ of the first iteration. By using nesting of subexpressions in combination with operator $+$, rather complicated regex can be constructed.\par
We refer to~\cite{Schmid2016, FreydenbergerSchmidJCSS} for a detailed definition of the semantics for regex; moreover, the following automaton representations of regex, that are central for this work, will also implicitly give a definition of the regex-languages. \par
For a class $R \subseteq \regex_{\Sigma, X}$, the \emph{matching problem} (\emph{for $R$}) is the problem to decide whether $w \in \lang(\alpha)$ for given $\alpha \in R$ and $w \in \Sigma^*$. 

\subsection{Memory Automata}

An $\NFA$ is a tuple $M = (Q, \Sigma, \delta, q_0, F)$ with a set $Q$ of states, a finite alphabet $\Sigma$, a start state $q_0$, a set $F$ of accepting states and a transition function $\delta : Q \times (\Sigma \cup \{\eword\}) \to \mathcal{P}(Q)$. Configurations of $M$ (on input $w$) are pairs $(q, u)$, where $q \in Q$ and $u$ is a (possibly empty) suffix of $w$; $(q_0, w)$ is the start configuration (of $M$ on $w$) and a configuration $(q, \eword)$ is accepting if $q \in F$. The transition relation $\vdash_M$ on the configurations is induced by $\delta$ in the natural way and a word $w$ is accepted (i.\,e., in the language $\lang(M)$ of $M$) if $(q_0, w) \vdash^*_M (q, \eword)$ with $q \in F$ (where $\vdash^*_M$ is the reflexive-transitive closure of $\vdash_M$). \par
In order to derive \emph{$k$-memory automata} ($\MFA(k)$, for short) from $\NFA$, we first define for every $k \in \mathbb{N}$ an alphabet $\memInstAlphabet_k = \{\open{x}, \close{x} \mid x \in [k]\}$ and for any alphabet $\Sigma$, we set $\Sigma_k = \Sigma \cup [k]$ and $\ekSigma = \Sigma_k \cup \{\eword\}$. Syntactically, an $\MFA(k)$ is an $\NFA = (Q, \Delta, \delta, q_0, F)$ with $\Delta = \ekSigma \cup \memInstAlphabet_{k}$; the semantics are as follows. Configurations of $\MFA(k)$ are tuples $(q, w, (u_1, r_1), \ldots, (u_k, r_k))$ with $q$ and $w$ being the current state and remaining input, respectively, and $(u_i, r_i)$ is the configuration of \emph{memory} $i$, for every $i \in [k]$, where $r_i \in \{\opened, \closed\}$ is the \emph{status} and $u_i \in \Sigma^*$ is the \emph{content} of memory $i$. The transition relation $\vdash_M$ is induced by $\delta$ as follows. We have $c \vdash_{M} c'$ if one of the following two cases apply:
\begin{enumerate}
\item $c = (q, v w, (u_1, r_1), \ldots, (u_k, r_k))$ and $c' = (p, w, (u'_1, r_1), \ldots, (u'_k, r_k))$ with
\begin{itemize}
\item $p \in \delta(q, x)$ with either ($x \in \eSigma$ and $v = x$) or ($x \in [k]$, $r_x = \closed$ and $v = u_x$), and,
\item for every $\ell \in [k]$, $r_{\ell} = \opened$ implies $u'_{\ell} = u_{\ell} v$, and $r_{\ell} = \closed$ implies $u'_{\ell} = u_{\ell}$.
\end{itemize}
\item $c = (q, w, (u_1, r_1), \ldots, (u_k, r_k))$ and $c' = (p, w, (u_1, r_1), \ldots, (u'_{\ell}, r'_{\ell}), \ldots, (u_k, r_k))$ with $p \in \delta(q, x)$ with $x = \open{\ell}$, $r'_{\ell} = \opened$ and $u'_{\ell} = \eword$, or with $x = \close{\ell}$, $r'_{\ell} = \closed$ and $u'_{\ell} = u_{\ell}$.
\end{enumerate}

Hence, intuitively speaking, we can consume either single symbols from the remaining input, or the whole content $u_i$ of a memory $i$ (although for this the memory must be closed, i.\,e., $r_x = \closed$), while everything that we consume from the input is appended to the content of every memory $j$ that is open (i.\,e., $r_{j} = \opened$). The special symbols $\open{i}$ and $\close{i}$ change the status of a memory $i$, without consuming anything from the remaining input.

The initial configuration of $M$ (on input $w$) is the configuration $(q_0, w, (\eword, \closed), \ldots, (\eword, \closed))$, a configuration $(q, \eword, (u_1, r_1), \ldots, (u_k, r_k))$ is an accepting configuration if $q \in F$, and $\lang(M)$ is the set of accepted inputs.\par
\begin{figure}
\centering
\scalebox{1}{\includegraphics{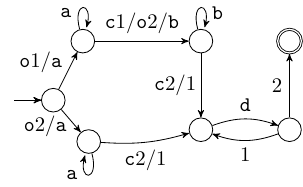}}
      	\caption{Example $\MFA(2)$.}
		\label{exampleMFAFigure}
\end{figure}
In the following, we shall denote $\MFA(k)$ by tuples $(Q, \Sigma, \delta, q_0, F)$ (i.\,e., we only explicitly state the ``actual'' terminal alphabet $\Sigma$). If the number $k$ of memories is not important, we also write $\MFA$. It will often be convenient to allow a slight abuse of notation and address memories with `names' rather than integers. For every $x \in \ekSigma \cup \memInstAlphabet_k$, transitions $p \in \delta(q, x)$ are called \emph{$x$-transitions}, and we also write $p \in \delta(q, x)$ as $(q, x) \to_{\delta} p$ (and also drop the subscript $\delta$ if it is clear from the context). An $x$-transition is a \emph{memory recall transition} if $x \in [k]$. A sequence $c_1, c_2, \ldots, c_m$ of configurations with $c_i \vdash_M c_{i + 1}$, $1 \leq i \leq m-1$, is called \emph{computation} (\emph{of $M$}), it is called a \emph{computation} (\emph{of $M$}) \emph{on input $w$}, if $c_1$ is the initial configuration on input $w$. For the sake of convenience, we also write computations as tuples $\vec{c} = (c_1, c_2, \ldots, c_m)$. For a configuration $(q, w, (u_1, r_1), \ldots, (u_k, r_k))$, $(u_1, \ldots, u_k)$ are the \emph{memory contents} and $(r_1, \ldots, r_k)$ the \emph{memory statuses}. \par
As usually done for $\NFA$, we also interpret $\MFA$ as directed graphs with vertices $Q$ and transitions as edge-labels; the start state is marked by an incoming arrow and accepting states are double-circled. We also label edges with several elements $x_1, x_2, \ldots, x_n \in \ekSigma \cup \memInstAlphabet_k$, separated by `\slash', in order to denote a sequence of transitions in a compact way.

\begin{example}\label{TMFAExample}
Consider the $M \in \MFA(2)$ illustrated in Fig.~\ref{exampleMFAFigure} (note that $\open{x}$ and $\close{x}$ are compressed to $\openshort{x}$ and $\closeshort{x}$, respectively). $M$ can either record a word from $\ta^+$ in memory $1$ and then a word $\tb^+$ in memory $2$ (this corresponds to the `upper branch'), or a word from $\ta^+$ in memory $2$ (this corresponds to the `lower branch'). Then, memory $1$ is recalled followed by reading $\td$, and these two steps can be repeated arbitrarily often. Finally, $M$ enters an accepting state by recalling memory $2$. Note that in the lower branch, memory $1$ is necessarily empty. As can be easily verified, $\lang(M) = \{\ta^n \tb^m (\ta^n \td)^k \tb^m \mid n, m, k \geq 1\} \cup \{\ta^n \td^m \ta^n \mid n, m \geq 1\}$; moreover, $\lang(M) = \lang (((x\{\ta^+\} y\{\tb^+\}) \altop y\{\ta^+\}) (x \td)^+ y)$. See also Example~\ref{regexExample} in the Appendix.
\end{example}

Given an $\MFA(k)$ and a word $w$, we can check whether $w \in \lang(M)$ as follows. In the graph that has the configurations $(q, u, (r_1, u_1), \ldots, (r_k, u_k))$ of $M$ on $w$ as vertices and edges given by the relation $\vdash_M$, we simply search for a path from the initial configuration to an accepting configuration. This directly yields the following trivial upper bound for the matching problem for memory automata. 

\begin{lemma}\label{MFAMembershipLemma}
Given $w \in \Sigma^*$ and $M = (Q, \Sigma, \delta, q_0, F) \in \MFA(k)$ with $|\delta| = \bigO(|Q|)$, we can decide $w \in \lang(M)$ in time $|Q||w|^{\bigO(k)}$.
\end{lemma}

\subsection{Memory Automata for Regex}\label{sec:RegexToMFA}

For $\alpha \in \regex_{\Sigma, X}$ with $X = [m]$, we transform $\syntaxtree(\alpha)$ into a directed, edge-labelled graph $\crudeAutomaton(\alpha)$. Every $t \in \nodes(\alpha)$ is replaced by nodes $t^{\instate}$, $t^{\outstate}$ if $\nodelabel(t) \in \{\altoptree, \plusoptree, [b], [x], [x\{\}] \mid b \in \Sigma_{\eword}, x \in X\}$, and by nodes $t^{\instate}$, $t^{\interstate}$, $t^{\outstate}$ if $\nodelabel(t) = \concoptree$. For every leaf $t \in \nodes(\alpha)$, we add an edge $(t^{\instate}, t^{\outstate})$, and for every non-leaf $t \in \nodes(\alpha)$, we do the following. 
\begin{itemize}
\item If $\nodelabel(t) = \concoptree$ and $r$ and $s$ are the left and right children of $t$, respectively, then we add edges $(t^{\instate}, r^{\instate})$, $(r^{\outstate}, t^{\interstate})$, $(t^{\interstate}, s^{\instate})$ and $(s^{\outstate}, t^{\outstate})$.
\item If $\nodelabel(t) = \altoptree$ and $r$ and $s$ are the left and right children of $t$, respectively, then we add edges $(t^{\instate}, r^{\instate})$, $(t^{\instate}, s^{\instate})$, $(r^{\outstate}, t^{\outstate})$ and $(s^{\outstate}, t^{\outstate})$.
\item If $\nodelabel(t) = \plusoptree$ and $r$ is $t$'s child, then we add edges $(t^{\instate}, r^{\instate})$, $(r^{\outstate}, t^{\outstate})$ and $(t^{\outstate}, t^{\instate})$.
\item If $\nodelabel(t) = [x\{\}]$ and $r$ is $t$'s child, then we add edges $(t^{\instate}, r^{\instate})$ and $(r^{\outstate}, t^{\outstate})$.
\end{itemize}
Every edge $(t^{\instate}, t^{\outstate})$ with $\subexpression(t) \in \eSigma \cup X$ is labelled by $\subexpression(t)$, every edge $(t^{\instate}, p)$ with $\nodelabel(t) = [x\{\}]$, for some $x \in X$, is labelled with $\open{x}$, and every edge $(t^{\outstate}, p)$ with $\nodelabel(t) = [x\{\}]$, for some $x \in X$, is labelled with $\close{x}$. Moreover, all other edges are labelled with $\eword$. We denote the set of nodes of $\crudeAutomaton(\alpha)$ by $\crudeAutomatonNodes(\alpha)$. \par
The graph $\crudeAutomaton(\alpha)$ is a directed graph with edge labels from $\emSigma \cup \memInstAlphabet_m$ and every vertex corresponds to a node of the syntax tree $\syntaxtree(\alpha)$ (see Fig.~\ref{Fig:mainExample} for an illustration). Consequently, $\crudeAutomaton(\alpha)$ can be interpreted both as an $\NFA$ over alphabet $\emSigma \cup \memInstAlphabet_m$ or as an $\MFA(m)$ over alphabet $\Sigma$. By $\canonicalNFA(\alpha)$ we denote the $\NFA$ obtained from $\crudeAutomaton(\alpha)$ by defining $t^{\instate}$ to be the initial state, $t^{\outstate}$ to be the only accepting state, where $t$ is the root of $\syntaxtree(\alpha)$, and the transition function to be represented by the edge-labels. Analogously, $\canonicalMFA(\alpha)$ is obtained by interpreting $\canonicalNFA(\alpha)$ as an $\MFA(m)$. 

In principle, the transformation of $\alpha$ into $\crudeAutomaton(\alpha)$ is the Thompson-construction that obtains an $\NFA$ from a regular expression. However, for our purpose it is convenient to keep this implicit correspondence between states and nodes of the syntax tree.

By consulting the formal definition of the syntax of regex in~\cite{Schmid2016, FreydenbergerSchmidJCSS}, and by considering that $\crudeAutomaton(\alpha)$ is a variant of the Thompson construction, the following is immediate.

\begin{proposition}\label{computeCrudeAutomatonProposition}
For every $\alpha \in \regex_{\Sigma, X}$, $\crudeAutomaton(\alpha)$ can be computed in time $\bigO(|\crudeAutomaton(\alpha)|) = \bigO(|\alpha|)$. Moreover, $\lang(\canonicalMFA(\alpha)) = \lang(\alpha)$. 
\end{proposition}

With Lemma~\ref{MFAMembershipLemma}, this means that the matching problem for $\regex_{\Sigma, X}$ can be solved in time $|\alpha||w|^{\bigO(|X|)}$.

\section{Regex with Bounded Active Variable Degree}\label{AVDSection}

\begin{figure}
\centering
\scalebox{0.8}{
\includegraphics{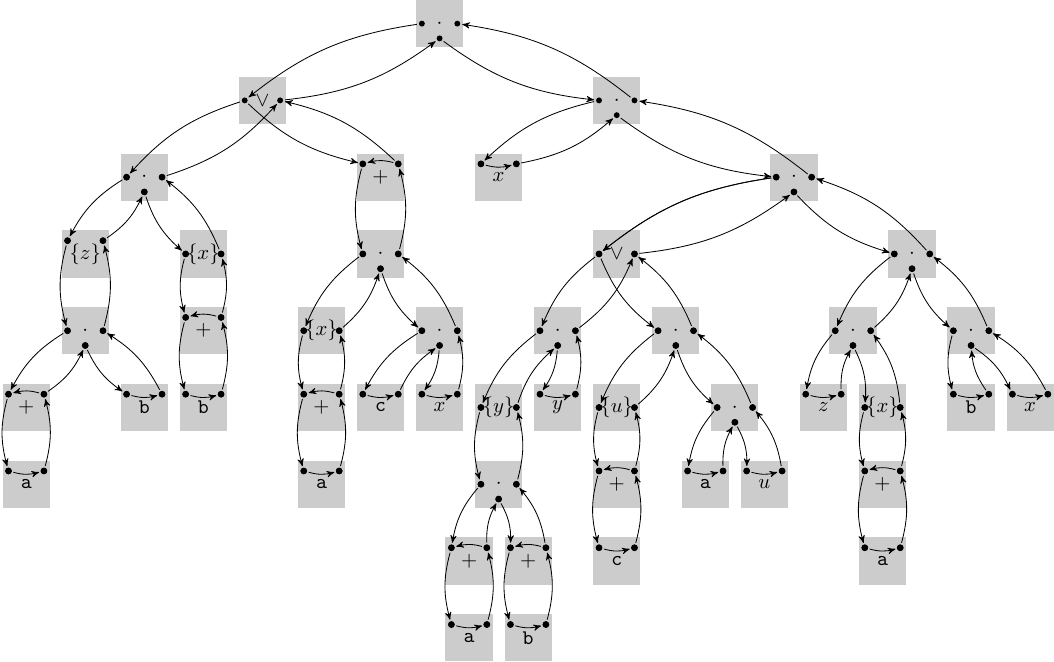}
}
\caption{$\crudeAutomaton(\alpha)$ for $\alpha = ((z\{\ta^+\tb\} x\{\tb^+\}) \altop (x\{\ta^+\}\tc x)^+) x ((y\{\ta^+\tb^+\}y)  \altop (u\{\tc^+\}\ta u)) z x\{\ta^+\}\tb x$. Every $t \in \syntaxtree(\alpha)$ is represented by a grey square labelled with $\nodelabel(t)$, which contains the nodes $t^{\instate}$ and $t^{\outstate}$ (and $t^{\interstate}$ if $\nodelabel(t) = [\cdot]$). In this way, the picture also implicitly shows $\syntaxtree(\alpha)$ (which also implicitly determines the omitted edge labels of $\crudeAutomaton(\alpha)$).}
\label{Fig:mainExample}
\end{figure}

In this section, we define efficiently matchable classes of regex by formalising the following observations. For regex of the form 
\begin{align*}
&(x_1\{\ta_1^+\} \altop x_2\{\ta_2^+\} \altop \ldots \altop x_n\{\ta_n^+\}) x_1 x_2 \ldots x_n \text{ or}\\
&(x_1\{\ta_1^+\}\:x_1)\: (x_2\{\ta_2^+\}\:x_2) \ldots (x_n\{\ta_n^+\}\:x_n)\,, 
\end{align*}
it is intuitively clear that they can be matched by a procedure that only has to store the value of one variable at a time, or, more formally, it is straightforward to construct an $\MFA(1)$. This is due to the fact that the $n$ different variables are independent. On the other hand, $$x_1\{\ta_1^+\}\: x_2\{\ta_2^+\} \ldots x_n\{\ta_n^+\}\: x_1 x_2 \ldots x_n$$ does not seem to have this nice property. A more complicated instance of this behaviour can be seen in Fig.~\ref{Fig:mainExample}: whenever variable $y$ or $u$ is defined, then the already defined variable $x$ will necessarily be redefined before it is recalled again. Moreover, as soon as $x$ is defined again, variables $y$ and $u$ are never recalled again. Consequently, when we reach a definition of $y$ or $u$, we can afford to forget $x$'s value, or, in the $\MFA$-perspective, it should be possible to use the memory for $x$ now for $y$ or $u$, and then later again for $x$, handling variables $x$, $y$ and $u$ with only one memory instead of three. For the example of Fig.~\ref{Fig:mainExample}, we can actually just rename $y$ and $u$ to $x$ and obtain an equivalent regex with only two variables. Unfortunately, the situation that variables can be reused, is not easily covered by a simple renaming of variables. For example, consider $\alpha = x\{\ta^+\} \tb x (y \tc y\{\tb^+\})^+ x\{\tb^+\} \ta x$ and $\beta = x\{\ta^+\} \tb x ((y\{\tb^+\}) \altop (z\{\tb^+\} \tc z)) y \ta x\{\tb^+\} \ta x$. For both these regex, it is again not necessary to store $x$'s value once we encounter another variable. However, renaming $y$ (and $z$ in the case of $\beta$) to $x$ produces non-equivalent regex, since it can happen that a former undefined occurrence of $y$ becomes a defined occurrence of $x$, e.\,g., $\alpha$ (with $y$ renamed to $x$) cannot generate $\ta \tb \ta \tc \tb \tb \ta \tb$, while $\beta$ (with $y$ renamed to $x$) cannot generate $\ta \tb \ta \tb \tc \tb \ta \tb \ta \tb$.\par
We now derive a complexity measure from these observations. For our definitions and the algorithm, we shall mainly rely on the automata-perspective, i.\,e., we work with $\canonicalNFA(\alpha)$ and $\canonicalMFA(\alpha)$. First, we define relations $\preRel \subseteq X \times \crudeAutomatonNodes(\alpha)$ and $\postRel \subseteq \crudeAutomatonNodes(\alpha) \times X$ as follows:\smallskip\\
\begin{tabular}{lcl}
$x \preRel q$ & $\Leftrightarrow$ & $\canonicalNFA(\alpha)$ can reach $q$ by reading a word $w$ with $|w|_{\open{x}} \geq 1$.\\
$q \postRel x$ & $\Leftrightarrow$ & starting in $q$, $\canonicalNFA(\alpha)$ can read a word $w x$ with $|w|_{\open{x}} = 0$.
\end{tabular}\smallskip\\
Intuitively speaking, $x \preRel q$ describes the situation that we can reach $q$ such that a definition for $x$ is reached along the way, which means that some memory is currently storing a value for $x$. Moreover, $q \postRel x$ means that from $q$ we can reach a point where the \emph{currently} stored value for $x$ is recalled, i.\,e., we can reach an $x$-transition without first resetting the memory for $x$ by an $\open{x}$-transition. Consequently, $x \preRel q \postRel x$ means that we can reach $q$ such that a memory is currently used for storing some value for $x$ and we cannot afford to lose this value.\par
For every $q \in \crudeAutomatonNodes(\alpha)$, the \emph{set of active variables} (\emph{for $q$}) is $$\actvarset(q) = \{x \mid x \preRel q \postRel x\}\,.$$ By $\avd(\alpha) = \max\{|\actvarset(t^{\instate})| \mid t \text{ has parent $t'$ with } \nodelabel(t') = [x\{\}], x \in X\}$, we denote the \emph{active variable degree} of $\alpha$. Finally, for every $k \in \mathbb{N}$, we define \emph{regex with active variable degree at most $k$} as $\regex^{\avd\leq k}_{\Sigma, X} = \{\alpha \in \regex_{\Sigma, X} \mid \avd(\alpha) \leq k\}$.\par
Coming back to our example of Fig.~\ref{Fig:mainExample}, we observe that $z \preRel t^{\instate} \postRel z$, where $t$ is the child node of $t'$ with $\subexpression(t') = [y\{\}]$. Moreover, $x \preRel t^{\instate} \notpostRel x$ and $u \notpreRel t^{\instate} \notpostRel u$. Thus, $\actvarset(t^{\instate}) = \{y, z\}$. In fact, we have $|\actvarset(s^{\instate})| \leq 2$ for every child node $s$ of some node $s'$ with $\nodelabel(s') \in \{[x\{\}], [y\{\}], [z\{\}], [u\{\}]\}$, and therefore $\avd(\alpha) = 2$.

\begin{lemma}\label{computeAVDLemma}
For $k \in \mathbb{N}$ and $\alpha \in \regex_{X, \Sigma}$, we can check $\avd(\alpha) \leq k$ in time $\bigO(|X||\alpha|^2)$.
\end{lemma}

\begin{proof}
For given $x \in X$ and $q \in \crudeAutomatonNodes(\alpha)$, we can check whether $x \preRel q$ in time $\bigO(|\alpha|)$. Indeed, this can be done by simply checking whether in $\canonicalNFA(\alpha)$ we can reach $q$ with a word that contains $\open{x}$. Analogously, we can check whether $x \preRel q$ in time $\bigO(|\alpha|)$, which means that the relations $\preRel$ and $\postRel$ can be computed in time $\bigO(|X||\alpha|^2)$.\par
Let $T = \{t^{\instate} \in \crudeAutomatonNodes(\alpha) \mid \nodelabel(t) = [x\{\}], x \in X\}$. In order to check $\avd(\alpha) \leq k$, we compute $\actvarset(t)$ for every $t \in T$. Since we have computed relations $\preRel$ and $\postRel$, this can be done in time $\bigO(|\crudeAutomatonNodes(\alpha)||X|) = \bigO(|\alpha||X|)$. Consequently, $\avd(\alpha) = \max\{|\actvarset(t)| \mid t \in T\}$ can be computed in total time $\bigO(|X||\alpha|^2)$.
\end{proof}

Let $t \in \nodes(\alpha)$ with $\nodelabel(t) = [x\{\}]$. Then every word $w$ that can be read by $\canonicalNFA(\alpha)$ starting in $t^{\instate}$ is a word that begins with $\open{x}$. Consequently, we conclude the following:

\begin{observation}\label{relationsObservation}
For every $t \in \nodes(\alpha)$ with $\nodelabel(t) = [x\{\}]$, we have $t^{\instate} \notpostRel x$.
\end{observation}

We are now ready to prove that regex with bounded active variable degree can be matched in polynomial-time.

\begin{theorem}\label{regexBoundedSCDTheorem}
For $k \in \mathbb{N}$, the $\regex^{\avd\leq k}_{\Sigma, X}$-matching problem can be solved in polynomial-time.
\end{theorem}

\begin{proof}
Let $\alpha \in \regex^{\avd\leq k}_{\Sigma, X}$ and let $w \in \Sigma^*$. We first compute the relations $\preRel$ and $\postRel$ as shown in the proof of Lemma~\ref{computeAVDLemma}, and we construct $\crudeAutomaton(\alpha)$. \par
In the following, we transform $\canonicalMFA(\alpha)$ into an equivalent $M'_{\alpha} \in \MFA(k)$. Intuitively speaking, we replace every state $q$ by states $(q, \memorylist)$ for every $\memorylist \in (X \cup \{\notinuse\})^k$. The idea of the \emph{memory lists} $\memorylist$ contained in the states is that they store the information which of the $m$ old memories of $\canonicalMFA(\alpha)$ are currently handled by which of the $k$ memories of $M'_{\alpha}$; more precisely, $\memorylist_q[\ell] = x$ if memory $\ell$ of $M'_{\alpha}$ currently plays the role of the old memory $x$ of $\canonicalMFA(\alpha)$, and $\memorylist_q[\ell] = \notinuse$ means that memory $\ell$ is currently ``not in use''. More precisely, we define $M'_{\alpha} = (Q', \Sigma, \delta', q'_0, F') \in \MFA(k)$ with $Q' = \{(q, \memorylist) \mid q \in \crudeAutomatonNodes(\alpha), \memorylist \in (X \cup \{\notinuse\})^k\}$, $q'_0 = (q_0, (\notinuse, \ldots, \notinuse))$ and $F' = \{(p_f, \memorylist) \mid \memorylist \in (X \cup \{\notinuse\})^k\}$, where $q'_0$ and $p_f$ is the initial and accepting state of $\canonicalMFA(\alpha)$. We note that $|Q'| = \bigO(|\alpha| |X|^k)$. \par
Next, we give a high-level description of the transitions of $M'_{\alpha}$. The general idea is that $M'_{\alpha}$ simulates the computation of $\canonicalMFA(\alpha)$. Whenever $\canonicalMFA(\alpha)$ uses some memory $x \in X$, $M'_{\alpha}$ chooses a memory $\ell$ with $\memorylist[\ell] = \notinuse$ and uses it in order to simulate memory $x$. This works fine as long as we do not run out of memories (i.\,e., $\canonicalMFA(\alpha)$ opens a memory, but $\memorylist[\ell] \neq \notinuse$ for all $\ell \in [k]$). We shall first define transitions such that $M'_{\alpha}$ can simulate $\canonicalMFA(\alpha)$ under the assumption that this problem does not occur. Later on, we will show how the transitions must be transformed and extended for the general case. 
\begin{itemize}
\item Any $b$-transition of $\canonicalMFA(\alpha)$ with $b \in \eSigma$ is just simulated without changing the memory list.
\item If $\canonicalMFA(\alpha)$ opens memory $x$, then $M'_{\alpha}$ opens some memory $\ell$ with $\memorylist[\ell] = \notinuse$ and sets $\memorylist[\ell] = x$.
\item If $\canonicalMFA(\alpha)$ closes memory $x$ or recalls memory $x$, then $M'_{\alpha}$ does the same with respect to memory $\ell$ with $\memorylist[\ell] = x$ and does not change the memory list.
\end{itemize}
With these transitions, $M'_{\alpha}$ can simulate $\canonicalMFA(\alpha)$ up to the situation where it performs an $\open{x}$-transition, but $\memorylist[\ell] \neq \notinuse$, for every $\ell \in [k]$, or a $\close{x}$- or $x$-transition, but $\memorylist[\ell] \neq x$, for every $\ell \in [k]$. \par
We next modify the transitions defined so far in the following way. For every $x \in X$, whenever $M'_{\alpha}$ moves from a state $(q, \memorylist_q)$ to a state $(p, \memorylist_p)$ such that $p \notpostRel x$, then all occurrences of $x$ in $\memorylist_p$ are replaced by $\notinuse$.\par
We claim that it is not possible now for $M'_{\alpha}$ to reach the situation that $\memorylist[\ell] = \memorylist[\ell'] = x$ for $1 \leq \ell < \ell' \leq k$ and $x \in X$. Initially, $\memorylist$ only stores $\notinuse$. The only way that $x$ is added to $\memorylist$ is that an $\open{x}$-transition is simulated. Since $\open{x}$-transitions are only triggered by states $t^{\instate}$ with $\nodelabel(t) = [x\{\}]$ and $t^{\instate} \notpostRel x$ (see Observation~\ref{relationsObservation}), a possible $x$-entry of $\memorylist$ will be replaced by $\notinuse$ before the next $\open{x}$-transition is to be simulated.\par
We now assume that $M'_{\alpha}$ reaches the situation that it tries to simulate an $\open{x}$-transition of $\canonicalMFA(\alpha)$, but $\memorylist[\ell] \neq \notinuse$, for every $\ell \in [k]$. Let $t^{\instate}$ with $\nodelabel(t) = [x\{\}]$ be the state that triggers this $\open{x}$-transition, let $q$ be the state this transition leads to, and let $y_1, y_2, \ldots, y_k \in X$ be the elements stored in $\memorylist$. We observe the following facts:
\begin{itemize}
\item $|\{y_1, y_2, \ldots, y_k\}| = k$ and $x \notin \{y_1, y_2, \ldots, y_k\}$: As shown above, $\memorylist[\ell] = \memorylist[\ell'] \neq \notinuse$ for $1 \leq \ell < \ell' \leq k$ is not possible, so $\memorylist$ stores $k$ distinct values. Furthermore, since $\nodelabel(t) = [x\{\}]$, we also have $t^{\instate} \notpostRel x$ (see Observation~\ref{relationsObservation}), which implies that a previous $x$-entry of $\memorylist$ would have been removed. Thus, $x \notin \{y_1, y_2, \ldots, y_k\}$.
\item $\{y_1, y_2, \ldots, y_k\} \subseteq \actvarset(q)$: For every $i \in [k]$, we have $y_i \preRel t^{\instate}$, since otherwise it is not possible for $y_i$ to be stored in $\memorylist$, and we also have $t^{\instate} \postRel y_i$, since otherwise $y_i$ cannot be in $\memorylist$. Moreover, $\actvarset(t^{\instate}) \subseteq \actvarset(q)$ holds due to the fact that there is just one transition from $t^{\instate}$ labelled with $\open{x}$.
\item $q \notpostRel x$: Since $x \preRel q$, $q \postRel x$ would imply $x \in \actvarset(q)$ and therefore $k + 1 \leq |\actvarset(q)| \leq \avd(\alpha) = k$, which is a contradiction. 
\end{itemize}
By definition, $q \notpostRel x$ means that there is no word $wx$ with $|w|_{\open{x}} = 0$ that can be read by $\canonicalNFA(\alpha)$ starting in $q$. This means that all possible further $x$-transitions are preceded by an $\open{x}$-transition. Consequently, if $M'_{\alpha}$ reaches the situation that it tries to simulate an $\open{x}$-transition of $\canonicalMFA(\alpha)$, but $\memorylist[\ell] \neq \notinuse$, for every $\ell \in [k]$, then we can simply ignore this $\open{x}$-transition of $\canonicalMFA(\alpha)$, i.\,e., we carry out an $\eword$-transition instead. Moreover, we will then necessarily also reach the situation that $M'_{\alpha}$ tries to simulate an $\close{x}$-transition of $\canonicalMFA(\alpha)$ (namely the one triggered by state $t^{\outstate}$), but $\memorylist[\ell] \neq x$, for every $\ell \in [k]$. We can also ignore this $\close{x}$-transition and just carry out an $\eword$-transition instead. \par
We only have to discuss the situation that $M'_{\alpha}$ tries to simulate an $x$-transition of $\canonicalMFA(\alpha)$, but $\memorylist[\ell] \neq x$, for every $\ell \in [k]$. Let us first assume that this happens when no $\open{x}$-transition has been simulated before. Then memory $x$ is empty, which means we can ignore the $x$-transition and just carry out an $\eword$-transition instead. Let us now assume that there has been an earlier $\open{x}$-transition triggered by some state $q$, and let us consider the $\open{x}$-transition that is the most recent one with respect to the $x$-transition to be simulated. There are two possibilities why this $\open{x}$-transition does not cause $x$ to be stored in $\memorylist$. The first one is that for the source state $p$ of this $\open{x}$-transition, we have $p \notpostRel x$. The second one is that in $q$ the memory list $\memorylist$ does not contain any occurrence of $\notinuse$, which, as explained above, also means that $p \notpostRel x$. However, $p \notpostRel x$ means again that there is no word $wx$ with $|w|_{\open{x}} = 0$ that can be read by $\canonicalNFA(\alpha)$ starting in $q$, which contradicts our assumption that the considered $\open{x}$-transition is the most recent one.\par
These considerations show that $M'_{\alpha}$ can simulate $\canonicalMFA(\alpha)$ and therefore $\lang(M_{\alpha}) = \lang(M'_{\alpha})$.\par
We can now check whether $w \in \lang(\alpha)$ by checking $w \in \lang(M'_{\alpha})$ in time $|Q'||w|^{\bigO(k)} = |\alpha| |X|^k |w|^{\bigO(k)}$ (see Lemma~\ref{MFAMembershipLemma}).
\end{proof}

The parameter $\avd$ has an obvious shortcoming: if for some $Y \subseteq X$, we have $x \preRel q$ for every $x \in Y$, then this only means that for every $x \in Y$ we can reach $q$ with $x$ defined, but not that it is possible to reach $q$ with \emph{all} $x \in Y$ defined at the same time. For example, $\alpha = ((x\{\ta^+\} y\{\tb^+\}) \altop z\{\tc^+\} \altop (x\{\tb^+\} u\{\tc^+\})) v\{\ta^+\} x y z u v$ has a maximum active variable degree of $|X| = 5$, while the maximum number of variables defined at the same time is only $3$ and we can easily define an $\MFA(3)$ for $\alpha$. Consequently, it seems that the active variable degree can be strengthened by extending the relation $\preRel$ to a relation of the form $\mathcal{P}(X) \times \crudeAutomatonNodes(\alpha)$ as follows: $\{y_1, y_2, \ldots, y_\ell\} \preRel q$ if and only if $\canonicalNFA(\alpha)$ can reach $q$ by reading a word $w$ with $|w|_{\open{y_i}} \geq 1$, for every $i \in [\ell]$. Then, we can define a \emph{strong active variable degree} by $\savd(\alpha) = \max\{|\actvarset(t^{\instate}) \cap Y| \mid Y \subseteq X, Y \preRel t^{\instate}, t \text{ has parent $t'$ with } \nodelabel(t') = [x\{\}]\}$.

\begin{theorem}\label{strongavdConphardnessTheorem}
Deciding whether $\savd(\alpha) \leq k$ for given $\alpha \in \regex_{\Sigma, X}$ and $k \in \mathbb{N}$ is $\conpclass$-hard.
\end{theorem}

\begin{proof}
We devise a reduction from the set cover problem:\smallskip\\
\begin{tabular}{ll}
\emph{Input}: &finite set $\mathcal{U}$, $\{B_1, B_2, \ldots, B_n\} \subseteq \mathcal{P}(\mathcal{U})$, and $k \in \mathbb{N}$.\\
\emph{Question}: & $\exists j_1, j_2, \ldots, j_k \in [n]$ with $\bigcup_{i \in [k]} B_{j_i} = \mathcal{U}$?
\end{tabular}\smallskip\\
Let $\mathcal{U} = \{x_1, x_2, \ldots, x_m\}$, $B_i = \{y_{i, 1}, y_{i, 2}, \ldots, y_{i, \ell_i}\} \subseteq \mathcal{U}$ and $k \in \mathbb{N}$ be an instance of the set cover problem. We transform this instance into an $\alpha \in \regex_{\Sigma, X}$ with $\Sigma = \{\tb\}$ and $X = \mathcal{U} \cup \{z\}$ as follows:
\begin{align*}
\beta_{i} &= y_{i, 1}\{\eword\} \: y_{i, 2}\{\eword\} \: \ldots \: y_{i, \ell_i}\{\eword\}\,, \hspace{1cm}  \text{for every $i \in [n]$,}\\
\gamma &= (\beta_{1} \altop \beta_{2} \altop \ldots \altop \beta_{n})^k\,,\\
\alpha &= z\{\eword\} \: \gamma \: \tb \: x_1 \: x_2 \ldots x_m \: z\,.
\end{align*}
We shall show that $\savd(\alpha) > n$ if and only if there are $j_1, j_2, \ldots, j_k \in [n]$ with $\bigcup_{i \in [k]} B_{j_i} = \mathcal{U}$. Let $t \in \crudeAutomatonNodes(\alpha)$ be the node that corresponds to the last occurrence of $\eword$ in $\gamma$. We observe that $\savd(\alpha) = \max\{|\actvarset(t^{\instate}) \cap Y| \mid Y \subseteq X, Y \preRel t^{\instate}\}$ and, since $\actvarset(t^{\instate}) = X$ obviously holds, $\savd(\alpha) > n$ is equivalent to $X \preRel t^{\instate}$. \par
For every $i \in [n]$, let $u_i = \open{y_{i, 1}}\close{y_{i, 1}}\,\open{y_{i, 2}}\close{y_{i, 2}} \ldots \open{y_{i, \ell_i}}\close{y_{i, \ell_i}}$. We note that $\lang(\canonicalNFA(\alpha)) = \{\open{z} \close{z} u_{j_1} u_{j_2} \ldots u_{j_k} \tb x_1 x_2 \ldots x_m z \mid j_i \in [n], i \in [k]\}$. Consequently, there are $j_1, j_2, \ldots, j_k \in [n]$ with $\bigcup_{i \in [k]} B_{j_i} = \mathcal{U}$ if and only if there is some $w = u \tb v \in \lang(\canonicalNFA(\alpha))$ with $|u|_{\open{y}} \geq 1$, for every $y \in X$. The second statement is, by definition of the relation $\preRel$, equivalent to $X \preRel t^{\instate}$. 
\end{proof}

In addition to the hardness of computing the strong active variable degree, it is also not entirely clear, how it could be used in the sense of Thm.~\ref{regexBoundedSCDTheorem}.\par 
By transforming the $\MFA(k)$ from Thm.~\ref{regexBoundedSCDTheorem} into a regex (see~\cite{Schmid2016} for details), we obtain the following corollary, which is worth mentioning, since deciding whether for a given $k$-variable regex there is an equivalent $(k-1)$-variable regex is undecidable (see~\cite{Freydenberger2013}).

\begin{corollary}
Every $\alpha \in \regex^{\avd\leq k}_{\Sigma, X}$ can be effectively transformed into a $\beta \in \regex_{\Sigma, X'}$ with $|X'| = k$ and $\lang(\alpha) = \lang(\beta)$.
\end{corollary}

\section{Memory-Deterministic Regex}

Considering $\NFA$ as a matching tool for classical regular expressions, their non-determinism could be considered harmless: the computation may branch in every step, but all parallel branches will differ only in their current states. Consequently, we can handle all possible parallel branches of an $\NFA$-computation by maintaining a set of \emph{active} states, which only causes a factor of $|Q|$ compared to the linear running time of a $\DFA$ (this can also be considered as determinising an $\NFA$ ``on-the-fly'').\footnote{Technically, we only get a factor $|Q|$ if $|\delta(q, a)|$ is constant for all $q$ and $a$, but this is the case for $\NFA$ obtained from regular expressions (see Thompson~\cite{Thompson1968} and Section~\ref{sec:RegexToMFA}).} Considering the fact that transforming regular expressions to $\DFA$ may cause exponential size blow-ups, while $\NFA$ of asymptotically the same size can be easily obtained, this additional factor of $|Q|$ is often acceptable.\par
For more complicated automata, e.\,g., with additional storage, it is often the case that the deterministic variant can be handled easily (but is of weak expressive power), while non-determinism causes undecidability or intractability. A typical way to approach this problem is to restrict the nondeterminism, hoping to find a more appealing balance between expressive power and complexity. \emph{Purely} deterministic $\MFA$ have been used in~\cite{FreydenbergerSchmidJCSS} to define \emph{deterministic} regex, which can be matched efficiently (in time $\bigO(|\Sigma||\alpha|^2 + k|w|)$, where $k$ is the number of variables), but, on the other hand, seem to be unnecessarily restricted if efficient matchability is our main concern: deterministic regex do not cover classical regular expressions, and the class of deterministic regex languages does not contain the class of regular languages (note that the latter statement is stronger than the former).\par
Our goal is to find a class of regex that properly extends classical regular expressions and for which the nondeterminism is only as powerful (and therefore as harmless) as for classical regular expressions (or $\NFA$). Since the variables of regex (or the memories of $\MFA$) are responsible for intractability, the main idea is to impose determinism on memories, but allow the harmless kind of nondeterminism observed in classical $\NFA$. Formalising this somewhat vague objective is not an easy task. We shall next substantiate this claim by demonstrating that even very mild forms of nondeterminism are sufficient to make the acceptance problem of $\MFA$ intractable. In particular, this result suggests that our goal cannot be achieved by local restrictions on a syntactic level. \par
Let $M = (Q, \Sigma, \delta, q_0, F) \in \MFA(k)$. A state $q \in Q$ is called \emph{deterministic}\label{deterministicDef} if, for every $x \in \ekSigma \cup \memInstAlphabet_k$, there is at most one $x$-transition for $q$, and it is called \emph{$x$-restricted} for an $x \in \ekSigma \cup \memInstAlphabet_k$, if the existence of an $x$-transition for $q$ implies that $q$ has no $y$-transitions for any $y \in (\ekSigma \cup \memInstAlphabet_k) \setminus \{x\}$. The $\MFA$ $M$ is \emph{deterministic} if all states are deterministic and, for every $x \in [k] \cup \memInstAlphabet_k \cup \{\eword\}$, all states are also $x$-restricted.\footnote{This definition slightly differs from the one in~\cite{FreydenbergerSchmidJCSS}, but yields the same model.}

\begin{theorem}\label{memdethardnessTheorem}
The acceptance problem for $\MFA(k)$ is $\npclass$-complete, even if the input $\MFA$ $M = (Q, \Sigma, \delta, q_0, F)$ have the following restrictions: (1) $\Sigma = \{\ta, \tb\}$, (2) $M$ has no $\eword$-transitions, (3) for every $q \in Q$ and $x \in [k] \cup \memInstAlphabet_k \cup \{\eword\}$, $q$ is $x$-restricted, (4) every $q \in Q$ is either deterministic or satisfies $|\delta(q, \ta)| = 2$.
\end{theorem}

\begin{proof}
We conduct a reduction from $\oneinthreethreesat$ without negated variables. To this end, let $C = (c_1, c_2, \ldots, c_m)$ be a set of clauses $c_i = \{y_{i, 1}, y_{i, 2}, y_{i, 3}\}$, $1 \leq i \leq m$, with $\bigcup^m_{i = 1} c_i = \{x_1, x_2, \ldots, x_n\}$. We define an $\MFA(2n)$, which, for every $i \in [n]$, has a memory $x_i$ and a memory $\overline{x_i}$. For every $i \in [n]$, we construct the component shown in Fig.~\ref{memdethardnessFigure}(a) and, for every $i \in [m]$, we construct the component shown in Fig.~\ref{memdethardnessFigure}(b) (note that $\open{x}$ and $\close{x}$ are compressed to $\openshort{x}$ and $\closeshort{x}$, respectively). In order to obtain $M$, we combine these components by joining some of their states (joining two states means that they will be the \emph{same} state in $M$). More precisely, we join every $p_{i}$ with $t_{i}$, $0 \leq i \leq n-1$, we join $t_n$ with $r_0$, and we join every $r_{i}$ with $s_{i}$, $1 \leq i \leq m-1$. Finally, we let $p_0$ be the start state and $s_m$ the only accepting state.\par
Next, we show that $C$ is $1$-in-$3$ satisfiable if and only if $(\ta \ta \tb)^{n} (\ta \tb)^{m}$ is accepted by $M$. In each computation, the $\MFA$ $M$ will initially read the word $(\ta \ta \tb)^{n}$ (which happens in the components shown in Fig.~\ref{memdethardnessFigure}(a)), and the $(2i)^{\text{th}}$ occurrence of $\ta$ will be stored in either memory $x_i$ or $\overline{x_i}$. Then, in the components shown in Fig.~\ref{memdethardnessFigure}(b), $M$ will read $m$ occurrences of $\tb$, where the $j^{\text{th}}$ occurrence of $\tb$ is directly preceded by the contents of the memories corresponding to clause $c_j$. The word consumed in this second part is $(\ta \tb)^{m}$ if and only if every clause contains exactly one memory that stores $\ta$.\par
Consequently, transforming a CNF-formula into an $\MFA$ as described above is a polynomial reduction from $\oneinthreethreesat$ without negated variables to the acceptance problem for $\MFA$. Moreover, the $\MFA$ obtained by this reduction satisfies the structural restrictions of the statement of the theorem (see Figures.~\ref{memdethardnessFigure}(a)~and~\ref{memdethardnessFigure}(b)).
\end{proof}

\begin{figure}
\begin{center}
\begin{tabular}{c c}
\resizebox{4.5cm}{!}
{\includegraphics{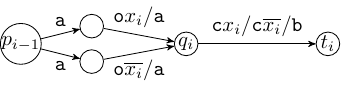}}
&
\resizebox{5.5cm}{!}
{\includegraphics{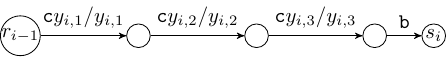}}\\
$(a)$ & $(b)$
\end{tabular}
\end{center}
\caption{Illustrations for the proof of Thm.~\ref{memdethardnessTheorem}.}
\label{memdethardnessFigure}
\end{figure}

The $\MFA$ of Thm.~\ref{memdethardnessTheorem} are quite restricted: the alphabet is binary and there are no $\eword$-transitions; moreover, each state has either just one outgoing transition or it has exactly $2$ $\ta$-transitions and not other transitions. Hence, they can be considered as being almost deterministic. In particular, the nondeterminism of the $\MFA$ of Thm.~\ref{memdethardnessTheorem} are especially restricted with respect to memories, since none of the non-deterministic branching points manipulate a memory. 
The actual problem seems to be that the undesired memory-nondeterminism does not present itself as a local nondeterministic choice, i.\,e., as two different transitions applicable in the same state that manipulate memories in different ways, but instead can arise much later in different computational branches that were created by a nondeterministic choice that seemingly does not cause memory-nondeterminism. This also suggests that a suitable restriction of the memory-nondeterminism can probably not be formulated as a local property for each separate state. Next, we define a property of $\MFA$ that covers our intuition of memory-determinism, but is rather complex in the sense that it depends on all possible computations of the $\MFA$, rather than on local properties of the transition function.\par

\subsection{Synchronised Memory Automata}

In the following, we shall show that the acceptance problem for so-called \emph{synchronised} $\MFA$ can be solved efficiently. In order to do this, we first need some algorithmic preliminaries.\par
Our computational model is the standard unit-cost RAM with logarithmic word size. We assume that $\NFA = (Q, \Sigma, \delta, q_0, F)$ (note that this includes $\MFA$, since they are, syntactically, $\NFA$) are given as directed graphs with vertices $Q$ (with special markers for the initial and accepting states) and $\delta$ is represented as edge-labels from $\Sigma$. Moreover, we assume that the out-degree is constant and for every vertex there is at most one symbol such that outgoing edges are labelled with this symbol. Hence, we can assume that the directed, edge-labelled graphs are represented by storing for each vertex the symbol for which outgoing edges exist, and also storing a set of the (constantly many) reachable vertices. From these assumptions, we directly conclude the following:
\begin{itemize}
\item The number of edges is $\bigO(|Q|)$, which also implies that $|\Sigma| = \bigO(|Q|)$ and therefore $|M| = \bigO(|Q|)$.
\item Given a vertex $q$ and $x \in \Sigma$, we can retrieve all $p$ with $p \in \delta(q, x)$ in constant time.
\item A breadth-first search can be performed in time $\bigO(|Q|)$.
\end{itemize}

\begin{remark}
The assumptions made above generally constitute a restriction to $\NFA$. They are nevertheless justified in our case, since they are all satisfied for the $\NFA$ and $\MFA$ obtained from regex, i.\,e., $\canonicalNFA(\alpha)$ and $\canonicalMFA(\alpha)$ (see Section~\ref{sec:prelim}). 
\end{remark}

For a sequence of memory instructions and $\eword$-symbols $C = (c_1, c_2, \ldots, c_n) \in (\memInstAlphabet_k \cup \{\eword\})^n$, we define $\compressedInstructions(C)$ as follows. For every $i \in [n]$, $c_i \in \compressedInstructions(C)$ if and only if $c_i \neq \eword$ and $c_i \in \{\open{x}, \close{x}\}$ implies that $c_j \notin \{\open{x}, \close{x}\}$, for every $j$ with $i < j \leq n$. Intuitively speaking, for every $x \in [k]$, we put only the very last occurrence (with respect to the sequence $C$) of any $\open{x}$ or $\close{x}$ into $\compressedInstructions(C)$ and ignore all the others. Obviously, applying transitions labelled with $c_1, c_2, \ldots, c_n$ in this order has the same effect as applying the memory instructions of $\compressedInstructions(c_1, c_2, \ldots, c_{n})$ in any order. A set $C \subseteq \memInstAlphabet_k$ is \emph{reduced} if, for every $x \in [k]$, $|\{\open{x}, \close{x}\} \cap C| \leq 1$. We note that for every $C' \in (\memInstAlphabet_k \cup \{\eword\})^n$, $\compressedInstructions(C')$ is reduced. \par
For every $q, p \in Q$ and reduced $C \subseteq \memInstAlphabet_k$, we write $(q, C, \eword) \to_{\contracted} p$ if there is a path from $q$ to $p$ of transitions labelled with $c_1, c_2, \ldots, c_n \in \memInstAlphabet_k \cup \{\eword\}$ such that $\compressedInstructions(c_1, c_2, \ldots, c_{n}) = C$. Furthermore, for every $x \in \kSigma$, we write $(q, C, x) \to_{\contracted} p$ if $(q, C, \eword) \to_{\contracted} p'$ with $p \in \delta(p', x)$.\par
We call $(q, C, x) \to_{\contracted} p$ a \emph{contracted transition}\footnote{It is discussed in the Appendix (Section~\ref{sec:appendixContractedTransitions}), why we cannot afford to actually compute all those contracted transitions.} and we set $\deltaContr(q, x) = \{p \mid \exists \text{ reduced } C \subseteq \memInstAlphabet_k: (q, C, x) \to_{\contracted} p\}$. In the following, let $\sigma = |\Sigma| + k$. A transition $(q, x) \to p$ is called \emph{consuming} if and only if $x \in \Sigma_k$.

\begin{lemma}\label{computeReachSetsLemma}
All sets $\deltaContr(q, x)$, $q \in Q$, $x \in \ekSigma$, can be computed in time $\bigO(|Q|^2\sigma)$.
\end{lemma}

\begin{proof}
For every $q \in Q$, we proceed as follows. First, we start a breadth-first search at $q$ that only considers non-consuming transitions and we build the corresponding tree (i.\,e., the breadth-first search tree), which requires time $\bigO(|Q|)$ (see the explanations from above). The states of this tree is exactly the set $\deltaContr(q, \eword)$. Next, for every $x \in \kSigma$, we initialise a set $A_x = \delta(q, x)$, and then we traverse the tree top-down (i.\,e., we repeat the breadth-first search) and every time we visit a state $p$, we add the set $\delta(p, x)$ to $A_x$ for every $x \in \kSigma$ (note that, as mentioned above, $|\delta(p, x)|$ is constant). After termination of this traversal, the sets $A_x$ are the sets $\deltaContr(q, x)$. These computations require time $\bigO(|Q|\sigma)$ for a fixed $q$, and therefore we need time $\bigO(|Q|^2\sigma)$ in order to compute all sets $\deltaContr(q, x)$.
\end{proof}

The \emph{contraction} of a computation $\vec{c}$ of some $M \in \MFA$ on some input $w$ is obtained by replacing every maximal sequence $(\vec{c}[i], \vec{c}[i + 1], \ldots, \vec{c}[j])$ whose corresponding computational steps are due to non-consuming transitions by $(\vec{c}[i], \vec{c}[j+1])$, or by $(\vec{c}[j])$, if $j = |\vec{c}|$. Obviously, if $\vec{c}$ is a contracted computation of some $\MFA$ $M$, then, for every $i \in [|\vec{c}|-1]$, $M$ has a contracted transition that can change $\vec{c}[i]$ to $\vec{c}[i + 1]$. \par
Two contracted computations $\vec{c}$ and $\vec{c'}$ for some $M \in \MFA$ are said to be \emph{synchronised} if and only if $\vec{c}[1] = \vec{c'}[1]$ and, for every $i \leq \min\{|\vec{c}|, |\vec{c'}|\}$, either $\vec{c}[i] = \vec{c'}[i]$ or $\vec{c}[i]$ and $\vec{c'}[i]$ only differ with respect to their states (i.\,e., their memory contents are the same, their memory statuses are the same, and their remaining inputs are the same). An $M \in \MFA$ is called \emph{synchronised} if any two contracted computations of $M$ on the same input are synchronised.\par

We next define some data-structures that are used by our algorithm for the acceptance problem of synchronised $\MFA$. Let $M = (Q, \Sigma, \delta, q_0, F)$ be a fixed $\MFA(k)$. Let $\synchmeminstrTable^M : Q \times \kSigma \to \mathcal{P}(\memInstAlphabet_k) \cup \{\nondef\}$ be a table such that, for every $q \in Q$ and $x \in \kSigma$, $\synchmeminstrTable^M[q, x]$ stores \emph{some} reduced $C \subseteq \memInstAlphabet_k$ for which there is a state $p$ and a contracted $x$-transition $(q, C, x) \to_{\contracted} p$ (or $\synchmeminstrTable^M[q, x] = \nondef$, if no such $C$ exists); note that for a fixed $M$, there are several valid possibilities for table $\synchmeminstrTable^M$. When $M$ is clear, we drop the superscript $M$.

\begin{lemma}\label{membershipPreprocessingLemma}
Table $\synchmeminstrTable$ can be computed in time $\bigO(|Q|^2\sigma)$.
\end{lemma}

\begin{proof}
We first initialise $\synchmeminstrTable$ by $\synchmeminstrTable[q, x] = \nondef$ for every $q \in Q$ and $x \in \kSigma$, which can be done in time $\bigO(|Q|^2\sigma)$. Then, for every $q \in Q$, we proceed as follows. First, we start a breadth-first search in $q$ that only considers non-consuming transitions and we build the corresponding tree (i.\,e., the breadth-first search tree), which requires time $\bigO(|Q|)$. Next, we compute a set $C_p \subseteq \memInstAlphabet_k$ for every state $p$ in this tree, such that $\deltaContr(q, C_p, \eword) \to_{\contracted} p$. This can be done as follows. We first initialise sets $C_p = \emptyset$ for every $p \in Q$. Then we traverse the tree top-down and every time we reach a new state $p'$ from a parent state $p$ via some $x$-transition (note that $x \in (\memInstAlphabet_k \cup \{\eword\})$), we set $C_{p'} = \compressedInstructions(C_p \cup \{x\})$ if $x \neq \eword$ and we set $C_{p'} = C_p$ otherwise. This requires time $\bigO(|Q|)$ (recall the explanations from the beginning of this subsection and also note that the computation of $\compressedInstructions(C_p \cup \{x\})$ only requires to add one element or to replace an element by a new one (this can be done by, e.\,g., a bit-vector implementation of the sets)). Finally, we traverse the tree another time top-down and for every visited state $p$ and every $x \in \kSigma$ with $\delta(p, x) \neq \emptyset$, we set $\synchmeminstrTable[q, x] = C_{p}$. This step requires time $\bigO(|Q|\sigma)$. \par
It can be easily verified that this procedure computes $\synchmeminstrTable$ correctly. Moreover, the total time required for these computations is $\bigO(|Q|^2 \sigma)$.
\end{proof}

Let $R : [k] \to \{\opened, \closed\}$ be a table and let $C \subseteq \memInstAlphabet_k$ be reduced. We define 
\begin{equation*}
(R \cup C)[i] = 
\begin{cases} 
\opened, & \text{if $\open{i} \in C$}\,,\\
\closed, & \text{if $\close{i} \in C$}\,,\\
R[i], & \text{else}\,.
\end{cases} 
\end{equation*}

The intuition here is that $R$ describes the memory statuses of some configuration of an $\MFA$ and $C$ describes some memory instructions (in the form of a reduced set of memory instructions as defined above). Then, $R \cup C$ simply describes the memory statuses after applying the instructions from $C$.

For a word $w \in \Sigma^*$, we store factors of $w$ by their start and end positions, which means that concatenating two adjacent factors or storing a factor in a program variable only requires constant time (note that since we assume a RAM with logarithmic word size, storing and manipulating positions of $w$ can be done in constant time). The \emph{longest common extension} data-structure $\lce_w$ is defined such that, for $i, j$, $1 \leq i < j \leq |w|$, $\lce_w(i, j)$ is the length of the longest common prefix of $w[i..|w|]$ and $w[j..|w|]$ (which can be retrieved in constant time). In particular, note that $\lce_w$ can be used in order to check in constant time whether a factor $w[i..j]$ is a prefix of a factor $w[i'..j']$, denoted by $w[i..j] \prefrel w[i'..j']$. In the following, we assume that we have $\lce_w$ at our disposal, which can be constructed in linear time (see, e.\,g.,~\cite{FischerHeun2006}).\par
Next, we define an algorithm (see Fig.~\ref{acceptanceAlgoFigure}) that, for a synchronised $M \in \MFA$ and a word $w$, decides whether or not $w \in \lang(M)$. Note that we use $\vec{\eword}$ and $\vec{\closed}$ as short hand for length-$k$ vectors each component of which is $\eword$ and $\closed$, respectively.

\begin{theorem}\label{matchingAlgoTheorem}
Let $M \in \MFA$ be synchronised and let $w \in \Sigma^*$. On input $(M, w)$, algorithm $\SynchMemAlgo$ decides whether or not $w \in \lang(M)$, and it can be implemented such that it has running-time $\bigO(|w||Q|^3\sigma)$.
\end{theorem}

\begin{figure}
\begin{algorithm}[H]
\SetAlgoNoEnd
\LinesNumbered
\SetSideCommentRight
\SetFillComment
\SetKwInOut{Input}{Input}
\SetKwInOut{Output}{Output}
\Input{Synchronised $M \in \MFA(k)$, $w \in \Sigma^*$.}
\Output{\textsf{Yes} if and only if $w \in \lang(M)$.}

For all $q \in Q$, $x \in \ekSigma$, compute sets $\deltaContr(q, x)$, compute table $\synchmeminstrTable$\label{preprocess}\;
$A := \{q_0\}$, $U := (u_1, \ldots, u_k) := \vec{\eword}$, $R := \vec{\closed}$, $v := w$, $b := 0$\;
\While{$b \leq |Q|$\label{mainLoop}}
{
      $\Lambda := \{x \in [k] \mid U[x] \prefrel v\} \cup \{v[1]\}$\label{LambdaLine}\;
	\If{$\exists q \in A, \exists x \in \Lambda : \deltaContr(q, x) \neq \emptyset$\label{mainLoopProperty}} 
		{ 
      	\textbf{if} $x \in \Sigma$ \textbf{then} $u := x$ \textbf{else} $u := U[x]$\label{setxLine}\;
		\textbf{if} $u = \eword$ \textbf{then} $b := b + 1$ \textbf{else} $b := 0$\label{countewordtrans}\;
		$R := R \cup \synchmeminstrTable[q,x]$\label{mainBodyStart}\;
		\For{$\ell \in [k]$ with $R[\ell] = \opened$\label{updateMemOne}}
	      	{
	      	$U[\ell] := U[\ell] u$\label{updateMemTwo}\;
		}
		$v := v[|u|+1..|v|]$\label{updatev}\;
		$A := \bigcup_{q \in A} \bigcup_{x \in \Lambda}\deltaContr(q, x)$\label{mainBodyEnd}\;
		\If{$(v = \eword) \wedge ((A \cup \bigcup_{q \in A} \deltaContr(q, \eword)) \cap F \neq \emptyset)$\label{acceptLineOne}}
		{
			\textbf{return} \textsf{Yes}\label{acceptLineTwo}\;
		}
		}
		\Else
			{
			\textbf{return} \textsf{No}\label{returnNoLine}\label{loopCondition}\;
			}
}
\textbf{return} \textsf{No}\;
\caption{$\SynchMemAlgo$}\label{synchMemAlgoLabelSecond}
\end{algorithm}
\caption{An algorithm for solving the acceptance problem for synchronised $\MFA$.}
\label{acceptanceAlgoFigure}
\end{figure}

\begin{proof}
All the following references to certain line numbers of the algorithm refer to Algorithm~\ref{synchMemAlgoLabelSecond} stated in Fig.~\ref{acceptanceAlgoFigure}.\par
We prove the correctness and the required running-time of the algorithm separately:\smallskip\\
\noindent\textbf{Correctness}: Let $m$ be the number of successful iterations of the main loop of the algorithm, i.\,e., iterations where the condition of Line~\ref{mainLoopProperty} is satisfied. By $\Lambda_i$, $A_i$, $U_i$, $R_i$ and $v_i$, we denote the values of the variables $\Lambda$, $A$, $U$, $R$ and $v$ at the beginning of the $i^{\text{th}}$ iteration, and by $q_i$ and $x_i$, we denote the state from $A_i$ and element from $\Lambda_i$, respectively, for which the condition of Line~\ref{mainLoopProperty} is satisfied at the $i^{\text{th}}$ iteration. \par
We note that at the $i^{\text{th}}$ iteration, the elements $A_i$, $U_i$, $R_i$ and $v_i$ encode a set of configurations 
\begin{equation*}
\mathcal{C}_{i} = \{(q, v_i, (U_i[1], R_i[1]), \ldots, (U_i[k], R_i[k])) \mid q \in A_i\} 
\end{equation*}
and, initially, $\mathcal{C}_{1}$ only contains the start configuration of $M$ on $w$.\smallskip\\ 
\noindent\emph{Claim $1$}: For every $i \in [m]$, $\mathcal{C}_{i}$ contains all configurations that can be reached from the initial configuration of $M$ on input $w$ by $(i-1)$ contracted transitions.\smallskip \\ 
\noindent\emph{Proof of Claim $1$}: The statement is obviously true for $i = 1$. Now let $i$ be arbitrary with $2 \leq i \leq m$ and assume that the statement holds for $i-1$. Let $c, c'$ be arbitrary configurations from $\mathcal{C}_{i-1}$ (possibly $c = c'$) with states $t_{i-1}$ and $t'_{i-1}$, respectively. \par
Since $M$ is synchronised and since $\mathcal{C}_{i-1}$ contains configurations with remaining input $v_{i-1}$ that can be reached from the initial configuration of $M$ on input $w$ by $(i-2)$ contracted transitions, it is not possible that there are contracted transitions $(t_{i-1}, C_{i-1}, y_{i-1}) \to_{\contracted} t_i$ and $(t'_{i-1}, C'_{i-1}, y'_{i-1}) \to_{\contracted} t'_i$ for $c$ and $c'$, respectively, such that $C_{i-1}$ and $C'_{i-1}$ have different effects on memory statuses $(R_{i-1}[1], \ldots, R_{i-1}[k])$ or $y_{i-1}$ and $y'_{i-1}$ cause different prefixes of $v_{i-1}$ to be consumed. Consequently, all contracted transitions applicable to configurations from $\mathcal{C}_{i-1}$ are such that they consume the same prefix of the remaining input $v_{i-1}$ (although the transitions may recall different memories, or one transition recalls a memory while the other consumes a single input symbol) and have the same effect on the memory statuses (although the actual memory instructions of the transitions might differ), i.\,e., they will lead to configurations that can only differ with respect to their states. Consequently, the set of all configurations that can be obtained by applying a contracted transition to a configuration from $\mathcal{C}_{i-1}$ can be obtained as follows: The new remaining input $v_i$ and the new memory configurations $(U_i[1], R_i[1]), \ldots, (U_i[k], R_i[k])$ (which are the same for all the configurations that can be obtained by applying a contracted transition to a configuration from $\mathcal{C}_{i-1}$) is obtained by carrying out one \emph{arbitrary} of these transitions (i.\,e., consuming a prefix of the input according to the transition and changing the memory statues $(U_{i-1}[1], R_{i-1}[1]), \ldots, (U_{i-1}[k], R_{i-1}[k])$ according to the memory instructions of the transition), while the new states can be obtained by collecting \emph{all} states that are reachable by \emph{any} contracted transition from some state of a configuration from $\mathcal{C}_{i-1}$ that is applicable (i.\,e., that consumes the first symbol from the remaining input or recalls a memory that stores a prefix of the remaining input).\par
We observe that in Line~\ref{LambdaLine}, we compute the set of all terminal symbols and memory recalls that can be part of an applicable transition, in Line~\ref{mainLoopProperty}, we check whether there is at least one applicable transition and if this is the case, we compute the new set $\mathcal{C}_{i}$ as described above in Lines~\ref{mainBodyStart}~to~\ref{mainBodyEnd}. Consequently, we can conclude that $\mathcal{C}_{i}$ contains all configurations that can be reached from a configuration of $\mathcal{C}_{i-1}$ by a contracted transition; thus, with the induction hypothesis, $\mathcal{C}_{i}$ contains all configurations that can be reached from the initial configuration of $M$ on input $w$ by $(i-1)$ contracted transitions.
\hfill (Claim $1$) $\square$\smallskip\\
By Claim~$1$, the algorithm searches all possible configurations that are reachable from the initial configuration of $M$ on input $w$. If among the current configurations there is an accepting one (i.\,e., the input is completely consumed and an accepting state is reached), then, due to Lines~\ref{acceptLineOne}~and~\ref{acceptLineTwo}, the algorithm terminates with output \textsf{Yes} (note that we also have to check whether an accepting state can be reached by non-consuming transitions). If there is no applicable transition for any of the current configuration from $\mathcal{C}_i$, then the condition in Line~\ref{mainLoopProperty} is not satisfied and therefore the algorithm terminates with output \textsf{No}. However, if the input has been completely consumed and the conditions in Line~\ref{acceptLineOne} is not satisfied, we cannot necessarily conclude that no accepting configuration is reachable, since there might be applicable transitions that recall empty memories and that eventually lead to an accepting configuration. Therefore, we proceed with the computation even though the remaining input is empty (note that in this case, Line~\ref{LambdaLine} is interpreted as $\Lambda = \{x \in [k] \mid U[x] = \eword\}$). \par
The only case not discussed so far is when we reach a loop of transitions that recall empty memories, but none of the traversed configurations are accepting (i.\,e., the remaining input is non-empty or none of the states is accepting). Due to the condition of Line~\ref{mainLoop}, which is not satisfied if the counter $b$ exceeds $|Q|$, and Line~\ref{countewordtrans}, in which $b$ is incremented if empty memories are recalled and reset to $0$ otherwise, the algorithm returns \textsf{No} if there is a sequence of at least $|Q| + 1$ consecutive transitions that recall empty memories. The following claim shows that this is correct (and therefore concludes the proof of correctness).\smallskip\\
\noindent\emph{Claim $2$}: If $b = |Q| + 1$ in Line~\ref{mainLoop}, then $w \notin \lang(M)$.\smallskip \\
\noindent\emph{Proof of Claim $2$}: A fundamental observation is that if for some $i$ and every $j$ with $i + 1 \leq j \leq i + d$, the set $\mathcal{C}_{j}$ is obtained from $\mathcal{C}_{j-1}$ by transitions that recall empty memories, and $(\bigcup^{i + d - 1}_{j = i} A_j) \cap A_{i + d} \neq \emptyset$, then a transition that recalls an empty memory is also applicable in iteration $i + d$ and therefore $A_{i + d + 1}$ is obtained from $A_{i + d}$ by transitions that recall an empty memory. This is due to the fact that in these iterations every empty memory stays empty and can therefore always be recalled if a state has such a transition. If $d$ is large enough, then, with respect to the visited states, there must exist a loop of transitions that recall an empty memory. More precisely, if $d > |Q|$, then, for every $j' \geq i + d$, there must be some state in $A_{j'}$ that allows a transition that recalls an empty memory. Therefore, $M$ is in an infinite loop of transitions that recall empty memories. We check whether this is happening by counting in $b$ (see Line~\ref{countewordtrans}) the number of such consecutive iterations caused by transitions that recall empty memories and interrupt the main loop (and return \textsf{No}) if $b$ properly exceeds $|Q|$.\par
If this happens with $v_i \neq \eword$, this is correct since $M$ cannot finish to consume its input in the loop of transitions that recall empty memories. Let us assume that in iteration $i$, i.\,e., when $M$ enters the loop, the remaining input is empty. In this case, we have to check whether an accepting state is reachable from an active one by performing transitions that recall empty memories, followed by a sequence of non-consuming transitions. If this is the case, then it must be possible to reach this state by at most $|Q|$ transitions that recall empty memories (followed by a sequence of non-consuming transitions). Consequently, it is sufficient to perform at most $|Q|$ more iterations and check whether accepting states can be reached from active states by non-consuming transitions, which is done in Lines~\ref{acceptLineOne}~and~\ref{acceptLineTwo}.
\hfill (Claim $2$) $\square$\smallskip\\
\noindent\textbf{Running-time}: We estimate the running-time by estimating the time required for the preprocessing in Line~\ref{preprocess}, the time required for one iteration of the main loop, and the maximum number of iterations of the main loop. According to Lemmas~\ref{computeReachSetsLemma}~and~\ref{membershipPreprocessingLemma}, the preprocessing can be done in time $\bigO(|Q|^2\sigma)$. Regarding an execution of the main loop, we observe that Line~\ref{countewordtrans} requires constant time and, since we can concatenate and compare factors in constant time (due to the $\lce_w$ data-structure), Lines~\ref{setxLine}~and~\ref{updatev} also require constant time. For executing each of the Lines~\ref{LambdaLine}~and~\ref{mainBodyStart}, as well as the complete loop of Lines~\ref{updateMemOne}~and~\ref{updateMemTwo}, time $\bigO(k)$ is sufficient (due to the $\lce_w$ data-structure), while the evaluation of the condition of Line~\ref{mainLoopProperty} requires time $\bigO(|Q|k)$. This leaves Line~\ref{mainBodyEnd}~and~\ref{acceptLineOne}, which require time $\bigO(|Q|^2k)$ and $\bigO(|Q|^2)$, respectively. Summing up, an execution of the main loop requires time $\bigO(|Q|^2k)$. Next, we note that the main loop is interrupted as soon as it is executed for $|Q| + 1$ times without reducing the remaining input; thus, it can be executed for $\bigO(|w||Q|)$ times in the worst case. We conclude that the total running-time of the algorithm is $\bigO(|Q|^2\sigma + |w||Q|^3k) = \bigO(|w||Q|^3\sigma)$.
\end{proof}

Technically, the class of $\alpha$ with synchronised $\canonicalMFA(\alpha)$ can be matched efficiently. Unfortunately, this class is of little use, since deciding membership to it is intractable.

\begin{figure}
\begin{center}
\begin{tabular}{c c c}
\resizebox{4cm}{!}
{\includegraphics{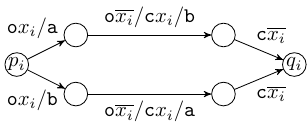}}
&
\resizebox{4cm}{!}
{\includegraphics{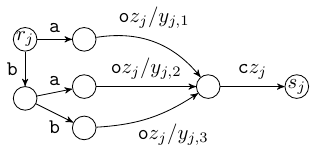}}
&
\resizebox{3.5cm}{!}
{\includegraphics{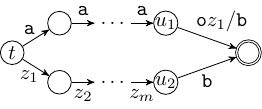}}\\
$(a)$ & $(b)$ & $(c)$
\end{tabular}
\end{center}
\caption{Illustrations for the proof of Thm.~\ref{synchronisedHardnesTheorem}.}
\label{memdethardnessFigureTwo}
\end{figure}

\begin{theorem}\label{synchronisedHardnesTheorem}
Deciding whether a given $\MFA$ is synchronised is $\conpclass$-hard.
\end{theorem}

\begin{proof}
We conduct a reduction from $\threesat$ to the problem of deciding whether a given $\MFA$ is not synchronised. Let $C = (c_1, c_2, \ldots, c_m)$ be a set of clauses $c_i = \{y_{i, 1}, y_{i, 2}, y_{i, 3}\}$, $1 \leq i \leq m$, with $\bigcup^m_{i = 1} c_i = \{x_1, \overline{x_1}, x_2, \overline{x_2}, \ldots, x_n, \overline{x_n}\}$. We define an $\MFA(2n+m)$ $M$, which, for every $i \in [n]$, has a memory $x_i$ and a memory $\overline{x_i}$, and, for every $j \in [m]$, a memory $z_j$. We first construct, for every $i \in [n]$, the component shown in Fig.~\ref{memdethardnessFigureTwo}(a), and, for every $j \in [m]$, the component shown in Fig.~\ref{memdethardnessFigureTwo}(b). These components are combined by joining states $q_i, p_{i+1}$, $1 \leq i \leq n-1$, joining states $q_n, r_1$, and joining states $s_j, r_{j+1}$, $1 \leq j \leq m-1$. Finally, we construct the component shown in Fig.~\ref{memdethardnessFigureTwo}(c), we join states $s_{m}, t$ and let $p_1$ be the start state (the only accepting state is shown in Fig.~\ref{memdethardnessFigureTwo}(c)). \par
Obviously, this reduction can be computed in polynomial-time and it only remains to prove its correctness.\par
If there is an input $w$, such that $M$ can reach both state $u_1$ and state $u_2$ by completely consuming $w$, then, since the transition with source $u_1$ stores $\tb$ in memory $z_1$, while the one with source $u_2$ does not change the content of memory $z_1$ (which must store $\ta$, since otherwise we cannot reach $u_1$ and $u_2$ with the same input word), there are non-synchronised computations with respect to input $w \tb$. On the other hand, since the part of $M$ that consists in the components of Fig.~\ref{memdethardnessFigureTwo}(a)~and~(b) is completely deterministic and the only nondeterminism of $M$ relies in the state $t$, the only non-synchronised computations must be due to an input that can lead $M$ into both states $u_1$ and $u_2$. Consequently, $M$ is non-synchronised if and only if it can reach both states $u_1$ and $u_2$ by consuming the same input. Obviously, such an input $w$ exists if and only if it is possible for $M$ to reach state $t$ with all memories $z_{j}$, $1 \leq j \leq m$, storing $\ta$. \par
Now we assume that $C$ is satisfiable and consider the following computation of $M$. If $x_i$ is assigned value \emph{true}, then, in the component of Fig.~\ref{memdethardnessFigureTwo}(a), we store $\ta$ in $x_i$ and $\tb$ in $\overline{x_i}$, and if $x_i$ is assigned value \emph{false}, then we store $\tb$ in $x_i$ and $\ta$ in $\overline{x_i}$. Since every clause contains a \emph{true} literal, it is possible to traverse the components of Fig.~\ref{memdethardnessFigureTwo}(b) in such a way that every memory $z_j$, $1 \leq j \leq m$, stores value $\ta$ when we reach state $t$. On the other hand, if we can reach $t$ with all memories $z_{j}$, $1 \leq j \leq m$, storing $\ta$, then for every $j \in [m]$, there is at least one memory among the memories $y_{j,1}, y_{j,2}, y_{j,3}$ that stores $\ta$, which, considering that the components of Fig.~\ref{memdethardnessFigureTwo}(a) force every pairs of memories $x_i, \overline{x_i}$ to store complementary values (with respect to the values $\ta$ and $\tb$), directly translates into a satisfying assignment of $C$.
\end{proof}

To achieve the goal stated at the beginning of this section, we formulate a slightly weaker, but sufficient criterion for the synchronisation property, which can be efficiently checked.

\subsection{Memory-Deterministic Regex}

Let $M = (Q, \Sigma, \delta, q_0, F) \in \MFA$. Recall that, for any $p, q \in Q$ and $x \in \kSigma$, $(q, C, x) \to_{\contracted} p$ means that we can reach $p$ from $q$ by reading some non-consuming symbols $c_{1}, c_{2}, \ldots, c_{m} \in \memInstAlphabet_k \cup \{\eword\}$ followed by reading $x$, such that $\compressedInstructions(c_{1}, c_{2}, \ldots, c_{m}) = C$. Moreover, for every $q \in Q$ and $x \in \Sigma_k$, the set $\deltaContr(q, x)$ contains all states $p$ such that $(q, C, x) \to_{\contracted} p$ holds for some $C$; in particular, $\deltaContr(q, x) = \emptyset$ means that from $q$ we cannot read $x$ after some sequence of non-consuming symbols.

We say that $q_1 \in Q$ and $q_2 \in Q$ are \emph{memory synchronised}, denoted by $q_1 \memSync q_2$, if the following conditions are satisfied.

\begin{itemize}
\item For every reduced $C_1, C_2 \subseteq \memInstAlphabet_k$, $x \in \kSigma$ and $p_1, p_2 \in Q$, $((q_1, C_1, x) \to_{\contracted} p_1) \wedge ((q_2, C_2, x) \to_{\contracted} p_2) \Rightarrow (C_1 = C_2)$,
\item For every $x \in [k]$, 
\begin{itemize}
\item $\deltaContr(q_1, x) \neq \emptyset \Rightarrow (\deltaContr(q_2, y) = \emptyset$ for every $y \in \Sigma_k \setminus \{x\})$.
\item $\deltaContr(q_2, x) \neq \emptyset \Rightarrow (\deltaContr(q_1, y) = \emptyset$ for every $y \in \Sigma_k \setminus \{x\})$.
\end{itemize}
\end{itemize}

Intuitively speaking, these conditions mean that if for both $q_1$ and $q_2$ we can reach consuming transitions reading the same symbol $x$, then the non-consuming transitions leading to these $x$-transitions must have the same effect on the memories. Moreover, if after some non-consuming transitions from $q_1$ we can reach a memory recall transition for $x \in [k]$, then $q_2$ is only allowed to reach consuming transitions that also recall the same memory $x$ (and vice versa). Also note that `$\memSync$' is not necessarily reflexive, i.\,e., there might be a state $q$ with $q \memSync q$.

We say that $q_1$ and $q_2$ are \emph{synchronised reachable}, denoted by $q_1 \samereach q_2$, if there is a word $w$ and synchronised computations $\vec{c}$ and $\vec{c'}$ of $M$ on input $w$ with $|\vec{c}| = |\vec{c'}| = m$ and the states of $\vec{c}[m]$ and $\vec{c'}[m]$ are $q_1$ and $q_2$, respectively. An $\MFA(k)$ $M = (Q, \Sigma, \delta, q_0, F)$ is \emph{memory-deterministic} (or an $\MDMFA(k)$, for short) if all states $q_1, q_2 \in Q$ that are synchronised reachable are also memory synchronised, i.\,e., for all $q_1, q_2 \in Q$, $q_1 \samereach q_2$ implies $q_1 \memSync q_2$.

\begin{lemma}\label{memDetSyncLemma}
Every $M \in \MDMFA$ is synchronised.
\end{lemma}

\begin{proof}
In this proof, we use the notation $c \vdash^{\contracted}_M c'$ or $c \vdash^{\contracted, *}_M c'$ in order to denote that $c'$ follows from $c$ (in one or several steps, respectively) in a contracted computation.\par
Let $M \in \MDMFA$ and let $\vec{c}$ and $\vec{c'}$ be two different contracted computations of $M$ on the same input. We shall show that $\vec{c}$ and $\vec{c'}$ are synchronised.\par
Let $i$, $1 \leq i \leq \max\{|\vec{c}|, |\vec{c'}|\} - 1$, be such that $\vec{c}[i]$ and $\vec{c'}[i]$ only differ in their states (or are identical), let $q_1$ and $q_2$ be the states of $\vec{c}[i]$ and $\vec{c'}[i]$, and let $(q_1, C_1, x_1) \to_{\contracted} p_1$ and $(q_2, C_2, x_2) \to_{\contracted} p_2$ be the contracted transitions responsible for $\vec{c}[i] \vdash^{\contracted} \vec{c}[i+1]$ and $\vec{c'}[i] \vdash^{\contracted} \vec{c'}[i+1]$, respectively. In particular, we observe that this means $q_1 \samereach q_2$. We can also note that since $\vec{c}[1] = \vec{c'}[1]$, such an $i$ exists. \par
If $x_1 \in [k]$, then we have $(q_1, C_1, x_1) \to_{\contracted} p_1$ with $x_1 \in [k]$; thus, the memory-determinism of $M$ implies that $\deltaContr(q_2, y) = \emptyset$ for every $y \in \Sigma_k \setminus \{x\}$, which implies that $x_1 = x_2$. Consequently, $(q_1, C_1, x_1) \to_{\contracted} p_1$ and $(q_2, C_2, x_1) \to_{\contracted} p_2$, which, since $q_1 \memSync q_2$, implies that $C_1 = C_2$. \par
If, on the other hand, $x_1 \in \Sigma$, then $x_2 \in [k]$ would be a contradiction to $q_1 \memSync q_2$, which implies that $x_2 \in \Sigma$. Thus, the contracted transitions $(q_1, C_1, x_1) \to_{\contracted} p_1$ and $(q_2, C_2, x_2) \to_{\contracted} p_2$ are both reading a symbol from the remaining input and since the remaining inputs of $\vec{c}[i]$ and $\vec{c'}[i]$ are the same, these two symbols must be the same, i.\,e., $x_1 = x_2$. As before, the memory-determinism now implies that $C_1 = C_2$. In both these cases, it follows that $\vec{c}[i+1]$ and $\vec{c'}[i+1]$ only differ with respect to their states (or are identical) and, since $q_1 \samereach q_2$, also $p_1 \samereach p_2$.\par
Consequently, by inductive application of this argument, it follows that $\vec{c}$ and $\vec{c'}$ are synchronised.
\end{proof}

We now define the class $\regex^{\mdet}_{\Sigma, X} = \{\alpha \in \regex_{\Sigma, X} \mid \canonicalMFA(\alpha) \in \MDMFA\}$ of \emph{memory-deterministic regex}, which, according to Thm.~\ref{matchingAlgoTheorem} and Lem.~\ref{memDetSyncLemma}, can be matched efficiently. 

\begin{corollary}
The $\regex^{\mdet}_{\Sigma, X}$-matching problem can be solved in time $\bigO(|w||\alpha|^3\sigma)$.
\end{corollary}

In order to substantiate this result, we discuss a more practically motivated example of a memory-deterministic regex in Section~\ref{sec:ExampleMDRegex} in the Appendix.

We shall next see that whether an $\MFA$ (and therefore a regex) is memory-deterministic can be checked in polynomial-time, which is the main benefit of the class $\regex^{\mdet}_{\Sigma, X}$. To this end, we first have to compute the relation $\memSync$. This is not entirely trivial, since it depends on contracted transitions, which we cannot afford to explicitly compute (see the remark in Section~\ref{sec:appendixContractedTransitions} of the Appendix). The idea is to first compute a data structure for answering queries of the form: ``given $q \in Q, x \in \kSigma, c \in \memInstAlphabet_k$, are there $p \in Q$ and reduced $C \subseteq \memInstAlphabet_k$ with $(q, C, x) \to_{\contracted} p$ and $c \in C$?'' These can be computed by analysing $\canonicalNFA(\alpha)$ and are sufficient to evaluate $\memSync$.

\begin{lemma}\label{computeMemSync}
The relation $\memSync$ can be computed in time $\bigO(|Q|^3\sigma k)$.
\end{lemma}

\begin{proof}
First, we compute the sets $\deltaContr(q, x)$, which, according to Lemma~\ref{computeReachSetsLemma}, requires time $\bigO(|Q|^2\sigma)$. We recall that for fixed $q_1, q_2 \in Q$, we have $q_1 \memSync q_2$, if the following conditions are satisfied.
 
\begin{itemize}
\item For every reduced $C_1, C_2 \subseteq \memInstAlphabet_k$, $x \in \kSigma$ and $p_1, p_2 \in Q$, $((q_1, C_1, x) \to_{\contracted} p_1) \wedge ((q_2, C_2, x) \to_{\contracted} p_2) \Rightarrow (C_1 = C_2)$.
\item For every $x \in [k]$, 
\begin{itemize}
\item $\deltaContr(q_1, x) \neq \emptyset \Rightarrow (\deltaContr(q_2, y) = \emptyset$ for every $y \in \Sigma_k \setminus \{x\})$.
\item $\deltaContr(q_2, x) \neq \emptyset \Rightarrow (\deltaContr(q_1, y) = \emptyset$ for every $y \in \Sigma_k \setminus \{x\})$.
\end{itemize}
\end{itemize}

The second property can be easily verified for all $q_1, q_2 \in Q$ by checking, for every $x \in [k]$ and $y \in \kSigma \setminus \{x\}$, whether $\deltaContr(q_1, x) \neq \emptyset$ and $\deltaContr(q_2, y) \neq \emptyset$, or whether $\deltaContr(q_2, x) \neq \emptyset$ and $\deltaContr(q_1, y) \neq \emptyset$. Since checks of the form $\deltaContr(q, x) = \emptyset$ can be done in constant time, this requires a total time of $\bigO(|Q|^2k\sigma)$. \par
In order to show how the first property can be checked for all $q_1, q_2 \in Q$, we define the following types of queries:\medskip\\
\begin{tabular}{ll}
$\mathsf{Q1}_{\in}(q, p, c)$: & does there exists a reduced set $C \subseteq \memInstAlphabet_k$ \\
& such that $(q, C, \eword) \to_{\contracted} p$ and $c \in C$?\\
$\mathsf{Q2}_{\in}(q, x, c)$: & does there exists a state $p \in Q$ and a reduced set $C \subseteq \memInstAlphabet_k$\\
& such that $(q, C, x) \to_{\contracted} p$ and $c \in C$?\\
\end{tabular}\medskip\\
With $\mathsf{Q1}_{\notin}(q, p, c)$ and $\mathsf{Q2}_{\notin}(q, x, c)$, we denote the variants of these queries with $c \notin C$ instead of $c \in C$.\par
Next, we discuss how queries $\mathsf{Q1}_{\in}(q, p, c)$ and $\mathsf{Q1}_{\notin}(q, p, c)$ can be evaluated. Let $c = \open{y}$ for some $y \in [k]$. Obviously, $\mathsf{Q1}_{\in}(q, p, \open{y})$ holds if in $M$ there is a path from $q$ to $p$ of transitions that are labelled with symbols from $\memInstAlphabet_k \cup \{\eword\}$, such that there is at least one $\open{y}$-transition in this path that is not followed by a $\close{y}$-transition. Checking the existence of such a path can be done as follows. We remove all $z$-transitions with $z \in \kSigma$, we replace all $z$-transition with $z \in \memInstAlphabet_k \setminus \{\open{y}, \close{y}\}$ by $\eword$-transitions, and we keep all original $\eword$-transitions. We have now obtained an $\NFA$ $M'$ over the alphabet $\{\open{y}, \close{y}\}$ with $\eword$-transitions. We declare $q$ to be the start state and $p$ to be the only accepting state. The $\NFA$ $M'$ accepts exactly the words over $\{\open{y}, \close{y}\}$ that correspond to the memory instructions for memory $y$ on a path of transitions that are labelled with symbols from $\memInstAlphabet_k \cup \{\eword\}$ and that leads from state $q$ to state $p$; moreover, constructing $M'$ can be done in time $\bigO(|Q|)$, and $M'$ has $|Q|$ states and at most $\bigO(|Q|)$ transitions. Consequently, the property to be checked holds if and only if $\lang(M') \cap \lang(N) \neq \emptyset$, where $N$ is an automaton for $\lang((\open{y} \altop \close{y})^* \open{y}^+)$ (note that $N$ has a constant number of states). Let $q_{0,N}$ be the start state and let $q_f$ be the only accepting state of $N$. By constructing the cross-product automaton of $M'$ and $N$ and checking reachability from $(q, q_{0, N})$ to $(p, q_f)$, we can check whether $\lang(M') \cap \lang(N) \neq \emptyset$ in time $\bigO(|Q|)$ (recall that there are $\bigO(|Q|)$ transitions). The argument for the case $c = \close{y}$ is analogous and in order to evaluate queries $\mathsf{Q1}_{\notin}(q, p, c)$, we let $N$ be an automaton for $\lang(((\open{y} \altop \close{y})^* \close{y}^+) \altop \eword)$.\par
Hence, we can evaluate the $\mathsf{Q1}_{\in}(q, p, c)$- and $\mathsf{Q1}_{\notin}(q, p, c)$-queries for \emph{all} $q, p \in Q$ and $c \in \memInstAlphabet_k$ in time $\bigO(|Q|^3k)$. Therefore, we assume in the following that $\mathsf{Q1}_{\in}$- and $\mathsf{Q1}_{\notin}$-queries can be answered in constant time. \par
By using $\mathsf{Q1}_{\in}$-queries, we can evaluate $\mathsf{Q2}_{\in}(q, x, c)$-queries as follows. We first check if there is a $p' \in \deltaContr(q, \eword)$ and a reduced $C \subseteq \memInstAlphabet_k$ with $(q, C, \eword) \to_{\contracted} p'$ and $c \in C$, which requires $\bigO(|Q|)$ $\mathsf{Q1}_{\in}$ queries. If this is the case, then we check whether there is a $p \in \delta(p', x)$, which can be done in constant time. Analogously, we can use $\mathsf{Q1}_{\notin}$-queries in order to evaluate $\mathsf{Q2}_{\notin}$-queries. This means that we can evaluate the $\mathsf{Q2}_{\in}(q, x, c)$- and $\mathsf{Q2}_{\notin}(q, x, c)$-queries for \emph{all} $q \in Q$, $x \in \kSigma$ and $c \in \memInstAlphabet_k$ in time $\bigO(|Q|^2\sigma k)$. As for $\mathsf{Q1}_{\in}$- and $\mathsf{Q1}_{\notin}$-queries, we shall now also assume that $\mathsf{Q2}_{\in}$- and $\mathsf{Q2}_{\notin}$-queries can be answered in constant time. \par
It remains to show how to check for every $q_1, q_2 \in Q$ whether the first property from above holds, i.\,e., the property:
\smallskip
\begin{quote}
For every reduced $C_1, C_2 \subseteq \memInstAlphabet_k$, $x \in \kSigma$ and $p_1, p_2 \in Q$, $((q_1, C_1, x) \to_{\contracted} p_1) \wedge ((q_2, C_2, x) \to_{\contracted} p_2) \Rightarrow (C_1 = C_2)$.
\end{quote}
\smallskip
To do this, we check for every $q_1, q_2 \in Q$, $x \in \kSigma$ and $c \in \memInstAlphabet_k$, whether $\mathsf{Q2}_{\in}(q_1, x, c)$ and $\mathsf{Q2}_{\notin}(q_2, x, c)$ hold, or whether $\mathsf{Q2}_{\notin}(q_1, x, c)$ and $\mathsf{Q2}_{\in}(q_2, x, c)$ hold. This requires time $\bigO(|Q|^2\sigma k)$. The total time required for all these computations is $\bigO(|Q|^3\sigma k)$.
\end{proof}

For checking whether a given $\MFA$ is memory-deterministic, we have to check whether there are states $q_1$ and $q_2$ such that $q_1 \samereach q_2$ and $q_1 \notmemSync q_2$. As shown by Lemma~\ref{computeMemSync}, checking whether or not $q_1 \memSync q_2$ for some $q_1, q_2 \in Q$ can be done in time $\bigO(|Q|^3\sigma k)$. The difficulty of checking whether $q_1 \samereach q_2$ is that this is not defined in terms of local properties of $q_1$ and $q_2$, or solely in terms of the structure of $M$. Thus, we define next a predicate that depends on structural properties of $M$ and that can be shown to be characteristic for $M$ being \emph{non}-memory-deterministic.

\begin{definition}
A triple $(q, p_1, p_2) \in Q^3$ is a \emph{non-synchronised branching triple}, denoted by the predicate $\nonSyncBranch(q, p_1, p_2)$, if $p_1$ and $p_2$ are not memory synchronised, i.\,e., $p_1 \notmemSync p_2$ (recall the definition of memory synchronised states from above) and either $q = p_1 = p_2$ or there are contracted transitions $(t_{1, i}, C_{i}, x_{i}) \to_{\contracted} t_{1, i+1}$ and $(t_{2, i}, C_i, x_{i}) \to_{\contracted} t_{2, i+1}$, $1 \leq i \leq m - 1$, such that $t_{1, 1} = t_{2, 1} = q$, $t_{1, m} = p_1$, $t_{2, m} = p_2$, and $t_{1, i} \memSync t_{2, i}$, for every $i \in \{2, 3, \ldots, m-1\}$.
\end{definition}

We first show that the existence of states $q, p_1, p_2$ with $\nonSyncBranch(q, p_1, p_2)$ characterises non-memory-determinism, and then we show how to check whether there are such states.

\begin{lemma}\label{branchingTriplesCharacterisationLemma}
$(\exists q_1, q_2 \in Q: q_1 \samereach q_2 \wedge q_1 \notmemSync q_2) \iff (\exists q, p_1, p_2 \in Q: \nonSyncBranch(q, p_1, p_2))$.
\end{lemma}

\begin{proof}
We start with the \emph{only if} direction and assume that there are $q_1, q_2 \in Q$ with $q_1 \samereach q_2$ and $q_1 \notmemSync q_2$. If $q_1 = q_2$, then we have that $q_1 \notmemSync q_1$ and therefore $\nonSyncBranch(q_1, q_1, q_1)$ holds; thus, we assume that $q_1 \neq q_2$ in the following. \par
By definition, $q_1 \samereach q_2$ implies that there is a word $w$ and synchronised computations $\vec{c}$ and $\vec{c'}$ of $M$ on input $w$ with $|\vec{c}| = |\vec{c'}| = m$ and the states of $\vec{c}[m]$ and $\vec{c'}[m]$ are $q_1$ and $q_2$, respectively. Since $\vec{c}[1] = \vec{c'}[1]$ and $\vec{c}[m] \neq \vec{c'}[m]$ (this follows from $q_1 \neq q_2$), there is some $j$ with $2 \leq j \leq m$ such that $\vec{c}[i] = \vec{c'}[i]$, $1 \leq i \leq j-1$, and $\vec{c}[j] \neq \vec{c'}[j]$. Since $\vec{c}[j]$ and $\vec{c'}[j]$ can only differ with respect to their states, we can conclude that $\vec{c}[j]$ has a state $t$ and $\vec{c'}[j]$ has a state $t'$, and $t \neq t'$ (note that if $j = m$, then $t = q_1$ and $t' = q_2$). Moreover, let $p$ be the common state of $\vec{c}[j-1]$ and $\vec{c'}[j-1]$.  \par
If $t \notmemSync t'$, then $\nonSyncBranch(p, t, t')$ holds. If, on the other hand, $t \memSync t'$, then, since $q_1 \notmemSync q_2$, there must be some $\ell$ with $j + 1 \leq \ell \leq m$ such that $\vec{c}[\ell]$ and $\vec{c'}[\ell]$ have states $s$ and $s'$ with $s \notmemSync s'$ (note that if $\ell = m$, then $s = q_1$ and $s' = q_2$), and, for every $i$ with $j+1 \leq i \leq \ell-1$, $\vec{c}[i]$ and $\vec{c'}[i]$ have some states $r_i$ and $r'_i$ with $r_i \memSync r'_i$. This implies that $\nonSyncBranch(p, s, s')$ holds.\par
In order to prove the other direction, we assume that, for some $q, p_1, p_2 \in Q$, $\nonSyncBranch(q, p_1, p_2)$ holds. By definition of $\nonSyncBranch$, this means that $p_1 \notmemSync p_2$; thus, it only remains to show $p_1 \samereach p_2$. Let $v$ be some word that leads $M$ to state $q$, i.\,e., on input $v$, $M$ can reach a configuration $(q, \eword, (u_1, r_1), \ldots, (u_k, r_k))$. Since $\nonSyncBranch(q, p_1, p_2)$ holds, there are contracted transitions $(t_{1, i}, C_{i}, x_{i}) \to_{\contracted} t_{1, i+1}$ and $(t_{2, i}, C_i, x_{i}) \to_{\contracted} t_{2, i+1}$, $1 \leq i \leq m - 1$, such that $t_{1, 1} = t_{2, 1} = q$, $t_{1, m} = p_1$, $t_{2, m} = p_2$, and $t_{1, i} \memSync t_{2, i}$, for every $i \in [m-1]$.
We now define a word $v'$ as follows: we initially set $v' = \eword$ and start $M$ on configuration $(q, \eword, (u_1, r_1), \ldots, (u_k, r_k))$ (i.\,e., the configuration that can be reached by $M$ on input $v$). Then we carry out the transitions $(t_{1, i}, C_{i}, x_{i}) \to_{\contracted} t_{1, i+1}$, $1 \leq i \leq m - 1$, one by one, and after applying $(t_{1, i}, C_{i}, x_{i}) \to_{\contracted} t_{1, i+1}$, we append $x_{i}$ to $v'$, if $x_{i} \in \Sigma$, and we append the current content of memory $x_{i}$ to $v'$, if $x_{i} \in [k]$. Obviously, $(q, v', (u_1, r_1), \ldots, (u_k, r_k)) \vdash^{\contracted, *} (p_1, \eword, (u'_1, r'_1), \ldots, (u'_k, r'_k))$, due to the transitions $(t_{1, i}, C_{i}, x_{i}) \to_{\contracted} t_{1, i+1}$, $1 \leq i \leq m - 1$, and, analogously, $(p_2, \eword, (u'_1, r'_1), \ldots, (u'_k, r'_k))$ is reachable from $(q, v', (u_1, r_1), \ldots, (u_k, r_k))$ via the states $t_{2, i+1}$, $1 \leq i \leq m - 1$ (in particular, note that these two computations are synchronised). Since, $(q_0, vv', (\closed, \eword), \ldots, (\closed, \eword)) \vdash^{\contracted, *} (q, v', (u_1, r_1), \ldots, (u_k, r_k))$, this directly implies that $p_1 \samereach p_2$.
\end{proof}

\begin{lemma}\label{CheckBranchingTriplesLemma}
Whether there are $q, p_1, p_2 \in Q$ with $\nonSyncBranch(q, p_1, p_2)$ can be checked in $\bigO(|Q|^5 + |Q|^3\sigma k)$.
\end{lemma}

\begin{proof}
We assume that the relation $\memSync$ is computed (which, according to Lemma~\ref{computeMemSync}, requires time $\bigO(|Q|^3\sigma k)$), and we assume that we have the sets $\deltaContr(q, x)$ at our disposal, which, according to Lemma~\ref{computeReachSetsLemma}, requires time $\bigO(|Q|^2\sigma)$. \par
Let $q, p_1, p_2 \in Q$ be fixed. Since $p_1 \memSync p_2$ implies that $\nonSyncBranch(q, p_1, p_2)$ does not hold, we only have to check $q, p_1, p_2 \in Q$ with $p_1 \notmemSync p_2$; thus, we assume that $p_1 \notmemSync p_2$ in the following. \par
We construct two $\NFA$ $M_{q, p_1}$ and $M_{q, p_2}$ from $M$ as follows. The automata $M_{q, p_1}$ and $M_{q, p_2}$ have the same states as $M$, and the only accepting state is $p_1$ and $p_2$, respectively; the start state of both $M_{q, p_1}$ and $M_{q, p_2}$ is $q$. Every $x$-transition with $x \in \memInstAlphabet_k$ of $M_{q, p_1}$ and $M_{q, p_2}$ is replaced by an $\eword$-transitions and all other transitions are left unchanged, i.\,e., $M_{q, p_1}$ and $M_{q, p_2}$ are $\NFA$ over the alphabet $\kSigma$ and with $\eword$-transitions. Next, we remove the $\eword$-transitions of $M_{q, p_1}$ and $M_{q, p_2}$ by setting $\delta(q, x) = \deltaContr(q, x)$, for every $q \in Q$ and $x \in \kSigma$. \par
If $\nonSyncBranch(q, p_1, p_2)$ holds, then there are contracted transitions $(t_{1, i}, C_{i}, x_{i}) \to_{\contracted} t_{1, i+1}$ and $(t_{2, i}, C_i, x_{i}) \to_{\contracted} t_{2, i+1}$, $1 \leq i \leq m - 1$, such that $t_{1, 1} = t_{2, 1} = q$, $t_{1, m} = p_1$, $t_{2, m} = p_2$, and $t_{1, i} \memSync t_{2, i}$, for every $i \in [m-1]$. In particular, this means that $x_1 x_2 \ldots x_{m-1} \in \lang(M_{q, p_1}) \cap \lang(M_{q, p_2})$. \par
On the other hand, if there is some $w \in \lang(M_{q, p_1}) \cap \lang(M_{q, p_2})$, then there are $(t_{1, i}, C_{1,i}, w[i]) \to_{\contracted} t_{1, i+1}$ and $(t_{2, i}, C_{2,i}, w[i]) \to_{\contracted} t_{2, i+1}$, $1 \leq i \leq |w|$, such that $t_{1, 1} = t_{2, 1} = q$, $t_{1, |w|+1} = p_1$, $t_{2, |w|+1} = p_2$ (note that these transitions are with respect to the original $\MFA$ $M$). If, in addition, $C_{1, i} = C_{2, i}$ and $t_{1, i} \memSync t_{2, i}$ for every $i$ with $1 \leq i \leq |w|$, then $\nonSyncBranch(q, p_1, p_2)$ holds. On the other hand, if one of these conditions is not satisfied, then let $s$, $1 \leq s \leq |w|-1$, be minimal such that $C_{1, i} = C_{2, i}$ and $t_{1, i} \memSync t_{2, i}$, $1 \leq i \leq s$. This implies that $t_{1, s+1} \notmemSync t_{2, s+1}$ or that $C_{1, s+1} \neq C_{2, s+1}$, which also implies $t_{1, s+1} \notmemSync t_{2, s+1}$. In both cases, $\nonSyncBranch(q, t_{1, s+1}, t_{2, s+1})$ holds. \par
Consequently, by checking $\lang(M_{q, p_1}) \cap \lang(M_{q, p_2}) \neq \emptyset$, for every $q, p_1, p_2 \in Q$ with $p_1 \notmemSync p_2$, we can check whether there are $q, p_1, p_2 \in Q$ such that $\nonSyncBranch(q, p_1, p_2)$ holds.\par
For fixed $q, p_1, p_2 \in Q$ with $p_1 \notmemSync p_2$, constructing $M_{q, p_1}$ and $M_{q, p_2}$ requires time $\bigO(|Q|)$. Next, we construct the cross-product automaton $M_{q, p_1, p_2}$ of $M_{q, p_1}$ and $M_{q, p_2}$, which accepts $\lang(M_{q, p_1}) \cap \lang(M_{q, p_2})$. Since both $M_{q, p_1}$ and $M_{q, p_2}$ have $|Q|$ states, $M_{q, p_1, p_2}$ has $|Q|^2$ states and therefore size $\bigO(|Q|^2)$. Checking whether $\lang(M_{q, p_1, p_2}) \neq \emptyset$ can therefore be done in time $\bigO(|M_{q, p_1, p_2}|) = \bigO(|Q|^2)$. We conclude that checking whether there are $q, p_1, p_2 \in Q$ such that $\nonSyncBranch(q, p_1, p_2)$ holds can be done in time $\bigO(|Q|^5)$. With the initial preprocessing, we get a total running time of $\bigO(|Q|^5 + |Q|^3\sigma k)$.
\end{proof}

These lemmas from above show how to check whether a given $\MFA$ is memory deterministic.

\begin{lemma}\label{checkMdetLemma}
Given $M \in \MFA$, we can decide whether $M \in \MDMFA$ in time $\bigO(|Q|^5 + |Q|^3\sigma k)$.
\end{lemma}

\begin{proof}
We first compute the relation $\memSync$ in time $\bigO(|Q^3|\sigma k)$ (see Lemma~\ref{computeMemSync}). Then we check whether there are $q, p_1, p_2 \in Q$ such that $\nonSyncBranch(q, p_1, p_2)$ holds in time $\bigO(|Q|^5 + |Q|^3\sigma k)$ (see Lemma~\ref{CheckBranchingTriplesLemma}). Since there are $q, p_1, p_2 \in Q$ such that $\nonSyncBranch(q, p_1, p_2)$ holds if and only if $M$ is not memory-deterministic (see Lemma~\ref{branchingTriplesCharacterisationLemma}), we have checked whether $M \in \MDMFA$ in total time $\bigO(|Q|^5 + |Q|^3\sigma k)$.
\end{proof}

Finally, Proposition~\ref{computeCrudeAutomatonProposition} and Lemma~\ref{checkMdetLemma} directly yield the following.

\begin{theorem}\label{checkMdetRegexTheorem}
Given $\alpha \in \regex_{\Sigma, X}$, we can decide whether $\alpha \in \regex^{\mdet}_{\Sigma, X}$ in time $\bigO(|\alpha|^5)$.
\end{theorem}

\section{Conclusions}

We presented two different approaches to efficient matching of regex. Since backreferences are the source of intractability, both approaches rely on restricting them somehow. The difference is that the concept of the active variable degree tries to reduce the number of required backreferences (although implicitly by re-using the memories of memory automata), while the concept of memory determinism relies on restricting how backreferences can be used (by imposing determinism for them). \par
The active variable degree can be considered as a complexity parameter that partitions the class of all regex into an infinite hierarchy of increasing matching complexity, i.\,e., a matching complexity that is exponential only in the active variable degree. In terms of efficiency and possible practical application, there are two main obstacles. Firstly, matching regex with active variable degree at most $k$ by constructing $\MFA(k)$ (as explained in Section~\ref{AVDSection}) only leads to acceptable running-times if $k$ is rather small. Secondly, the running times are of the form $|\alpha||w|^{\bigO(\avd(\alpha))}$, which is rather problematic under the reasonable assumption that $w$ is large and $\alpha$ is small. \par
From a language theoretical point of view, the hierarchy induced by the active variable degree is related to the natural hierarchy induced by the number of backreferences or variables that are necessary to describe the language of some regex $\alpha$ (called $\minVarPara(\alpha)$ in the following). Since regex with $k$ variables can be matched in time $|\alpha||w|^{\bigO(k)}$, the parameter $\minVarPara(\alpha)$ has a similar meaning as the active variable degree. However, its algorithmic application is questionable, since it is not computable (see~\cite{Freydenberger2013}). In this regard, the active variable degree can be interpreted as a computable upper bound for $\minVarPara(\alpha)$. \par
Memory deterministic regex, on the other hand, have the nice property that they can be matched in time $\bigO(|w||\alpha|^3\sigma)$, i.\,e., in a running-time of the form $p(|\alpha|)|w|$ for a polynomial $p$, or in linear time if measured in data-complexity. Their disadvantage is that checking memory determinism requires time $\bigO(|\alpha|^5)$ and even though their expressive power properly extends the one of deterministic regex (see~\cite{FreydenbergerSchmidJCSS}), their relevance is unclear in a practical context (however, see Section~\ref{sec:ExampleMDRegex} in the Appendix for a discussion of a more practically motivated memory-deterministic regex). However, the fact that they properly extend classical regular expressions (unlike the deterministic regex from~\cite{FreydenbergerSchmidJCSS}), which are without doubt of high practical relevance, justifies some hope that they could be used for practical purposes.\par
A natural extension of memory determinism would be to allow the contents and statuses of only few specific memories to differ in different computational branches, while all others must be synchronised. The concept of memory-determinism could be extended accordingly, e.\,g., if only memories $Y \subseteq [k]$ are allowed to be nondeterministic, then ``memory-determinism'' means that $q_1 \samereach q_2$ implies $q_1 \memSync q_2$, but the definition of $\memSync$ is changed such that $q_1 \memSync q_2$, if the following conditions are satisfied.
 
\begin{itemize}
\item For every reduced $C_1, C_2 \subseteq \memInstAlphabet_k$, $x \in \kSigma$ and $p_1, p_2 \in Q$, 
\begin{align*}
&((q_1, C_1, x) \to_{\contracted} p_1) \wedge ((q_2, C_2, x) \to_{\contracted} p_2) \:\: \Rightarrow \\
&(C_1 \setminus \{\open{j}, \close{j} \mid j \in Y\} = C_2 \setminus \{\open{j}, \close{j} \mid j \in Y\})\,.
\end{align*}
\item For every $x \in [k] \setminus Y$, 
\begin{itemize}
\item $\deltaContr(q_1, x) \neq \emptyset \Rightarrow (\deltaContr(q_2, y) = \emptyset$ for every $y \in \Sigma_k \setminus \{x\})$.
\item $\deltaContr(q_2, x) \neq \emptyset \Rightarrow (\deltaContr(q_1, y) = \emptyset$ for every $y \in \Sigma_k \setminus \{x\})$.
\end{itemize}
\end{itemize}
\par

Both the active variable degree as well as the concept of memory determinism allow obvious improvements, i.\,e., the \emph{strong} active variable degree and \emph{synchronised} memory automata, respectively. However, as substantiated by $\conpclass$-hardness results (Theorems~\ref{strongavdConphardnessTheorem}~and~\ref{synchronisedHardnesTheorem}), these improvements lead to intractability and are therefore not investigated further.  \par
Last but not least, for all our concepts and results we heavily rely on memory automata, which further substantiates their usefulness as a matching tool for regex. In particular, we believe that it would be worthwhile to implement a regex matching tool that is based on memory automata, which would also be a first step in implementing the approaches developed in this work.


\newpage

\appendix

\section{A Detailed Definition of Regex}

Here, we define the concept of regular expressions with backreferences in more detail and also give a sound definition of their semantics.\par
Let $X$ denote a finite set of \emph{variables} (as a convention, we normally use symbols like $x, y, z, x_1, x_2, x_3, \ldots$ to denote variables). The set $\regex_{\Sigma, X}$ of \emph{regular expressions with backreferences} (\emph{over $\Sigma$ and $X$}), also denoted by \emph{regex}, for short, is recursively defined as follows:
\begin{itemize}
\item For every $a \in \eSigma$, $a \in \regex_{\Sigma, X}$ and $\var(a) = \emptyset$.
\item For every $\alpha,\beta\in \regex_{\Sigma, X}$. 
\begin{itemize}
\item $(\alpha \cdot \beta) \in \regex_{\Sigma, X}$ and $\var((\alpha \cdot \beta)) = \var(\alpha)\cup\var(\beta)$,
\item $(\alpha \altop \beta) \in \regex_{\Sigma, X}$ and $\var((\alpha \altop \beta)) = \var(\alpha)\cup\var(\beta)$,
\item $(\alpha)^+ \in \regex_{\Sigma, X}$ and $\var((\alpha)^+) = \var(\alpha)$.
\end{itemize}
\item For every $x \in X$, $x \in \regex_{\Sigma, X}$ and $\var(x) = \{x\}$.
\item For every $\alpha \in \regex_{\Sigma, X}$ and $x \in X \setminus \var(\alpha)$, $x\{\alpha\} \in \regex_{\Sigma, X}$ and $\var(x\{\alpha\}) = \var(\alpha) \cup \{x\}$.
\end{itemize}
For $\alpha \in \regex_{\Sigma, X}$, we use $\alpha^*$ as a shorthand form for $\alpha^+ \altop \eword$, and we usually omit the operator `$\cdot$', i.\,e., we use juxtaposition. If the underlying alphabet $\Sigma$ and set $X$ of variables is negligible or clear from the context, we also denote the set of regex by $\regex$. In a regex, we call an occurrence of symbol $x \in X$ a \emph{reference to variable $x$} and a subexpression of the form $x\{\alpha\}$ a \emph{binding of variable $x$}; if we just talk about (\emph{occurrences of}) \emph{variables}, then we refer to a reference or a binding. We note that the subset of $\regex_{\Sigma, X}$ that can be created by the first two points from above is exactly the set of regular expressions over $\Sigma$, which we also call \emph{classical} regular expressions, in order to distinguish them from the regex defined above, and which we denote by $\regexp_{\Sigma}$. \par
There are several ways of how the semantics of regex can be formally defined (see the discussion in~\cite{FreydenbergerSchmid2017arxiv, FreydenbergerSchmidJCSS}). We use a rather simple one (introduced in~\cite{Schmid2016} and also used in~\cite{FreydenbergerSchmid2017, FreydenbergerSchmidJCSS}) that only relies on so-called \emph{ref-words} and \emph{ref-languages}, and classical regular expressions. Let $\paraAlphabet = \{\lpara_x, \rpara_x \mid x \in X\}$ be a set that contains a special pair of brackets for each variable. For an $\alpha \in \regex$, the \emph{ref-version} of $\alpha$ is the classical regular expression $\alpha_{\refexp} \in \regexp_{\Sigma \cup X \cup \paraAlphabet}$ that is obtained from $\alpha$ by recursively replacing each variable binding $x\{\beta\}$ by $\lpara_x \beta \rpara_x$. The ref-language of $\alpha$ is then defined by $\reflang(\alpha) = \lang(\alpha_{\refexp})$. For a \emph{ref-word} $w \in \reflang(\alpha)$, the \emph{dereference} $\deref(w)$ is obtained by recursively replacing every occurrence of a variable $x$ by $\beta$, where $\lpara_x \beta \rpara_x$ is the next matching pair of $\lpara_x, \rpara_x$ brackets to the left of this occurrence of $x$, or $\varepsilon$ if no such pair exists. Note that, by definition of regex and its ref-versions, every $w \in \reflang(\alpha)$ is well-formed with respect to each individual pair of brackets $\lpara_x, \rpara_x$, and it is impossible that some $x$ occurs between matching pairs of $\lpara_x, \rpara_x$. We extend $\deref$ from ref-words to ref-languages in the natural way. Finally, for every $\alpha \in \regex$, we define $\lang(\alpha) = \{\deref(w) \mid w \in \reflang(\alpha)\}$.\par
The following example illustrates these definitions.

\begin{example}\label{regexExample}
Consider $\alpha = ((x\{\ta^+\} y\{\tb^+\}) \altop y\{\ta^+\}) (x \td)^+ y \in \regex_{\Sigma, X}$ with $\Sigma = \{\ta, \tb, \td\}$, $X = \{x, y\}$, and the corresponding $\alpha_{\refexp} = ((\lpara_x \ta^+ \rpara_x \lpara_y\tb^+\rpara_y) \altop (\lpara_y\ta^+\rpara_y)) (x \td)^+ y \in \regexp_{\Sigma \cup X \cup \paraAlphabet}$ with $\paraAlphabet = \{\lpara_x, \rpara_x, \lpara_y, \rpara_y\}$. We note that, e.\,g., $w_1 = \lpara_x \ta \ta \rpara_x \lpara_y\tb\rpara_y x \td y \in \reflang(\alpha)$ and $w_2 = \lpara_y \ta \ta \ta \rpara_y x \td x \td y \in \reflang(\alpha)$, which implies $\deref(w_1) = \ta \ta \tb \ta \ta \td \tb \in \lang(\alpha)$ and $\deref(w_2) = \ta \ta \ta \td \td \ta \ta \ta \in \lang(\alpha)$. In particular, due to the presence of an alternation operator, $\reflang(\alpha)$ can contain words with occurrences of variable $x$ that are not preceded by a pair of brackets $\lpara_x, \rpara_x$ and are therefore replaced by $\varepsilon$ by the dereference-function $\deref$. As can be easily verified, the language described by $\alpha$ is $\{\ta^n \tb^m (\ta^n \td)^k \tb^m \mid n, m, k \geq 1\} \cup \{\ta^n \td^k \ta^n \mid n, k \geq 1\}$. 
\end{example}

\section{A Remark on Contracted Transitions of Memory Automata}\label{sec:appendixContractedTransitions}

For our definitions and results, the concept of contracted transitions are central. Computing all those contracted transitions and replacing the normal transitions by the contracted ones could be considered as making an $\MFA$ $\eword$-free. More precisely, we could compute for every $q \in Q$ and $x \in \kSigma$ the set $\mathcal{CT}(q, x)$ of all $(p, C)$ such that $q \to_{\contracted} (p, C, x)$. However, we observe that $|\mathcal{CT}(q, x)|$ is not necessarily polynomially bounded. For example, if, for every $x \in \memInstAlphabet_k$, and $j$, $1 \leq j \leq k$, we have $(p_{j-1}, x) \to p_{j}$, then it can be easily verified that for every reduced $C \in \memInstAlphabet_k$, we have $p_0 \to_{\contracted} (p_k, C, x)$, which implies that $|\mathcal{CT}(q, x)| \geq 2^k$. On the other hand, changing the model of $\MFA$ such that it only allows what we denoted by contracted transitions is problematic, since then these contracted transitions need to be computed in the transformation from regex to $\MFA$. This also shows and justifies why our concept of memory determinism is complicated on a technical level.

\section{Example of a Memory-Deterministic Regex}\label{sec:ExampleMDRegex}

\begin{figure}

\resizebox{12.5cm}{!}
{
\begin{tikzpicture}[node distance=10mm,on grid,>=stealth',auto, 
state/.style={circle,draw=black,inner sep=0pt,minimum size=4mm}]
\node[state]  (start)  at (0,0)  {$$};
\node[state, draw=none]  (prestart)  [below=of start]  {$$};
\node[state]  (r1)  [right=of start, xshift=0.5cm]  {$r$};
\node[state]  (r2)  [right=of r1, xshift=0.5cm]  {$$};
\node[state]  (n1)  [right=of r2]  {$$};
\node[state]  (s1)  [right=of n1, yshift=1cm]  {$$};
\node[state]  (s2)  [right=of s1]  {$$};
\node[state]  (s3)  [right=of s2]  {$$};
\node[state]  (s4)  [right=of n1, yshift=-1cm]  {$$};
\node[state]  (s5)  [right=of s4]  {$$};
\node[state]  (s6)  [right=of s5]  {$$};
\node[state]  (s7)  [below=of s3]  {$$};
\node[state]  (c1)  [right=of s7, yshift=1cm]  {$$};
\node[state]  (c2)  [right=of c1]  {$$};
\node[state]  (c3)  [right=of c2]  {$$};
\node[state]  (c4)  [right=of c3]  {$$};
\node[state]  (c5)  [right=of s7, yshift=-1cm]  {$$};
\node[state]  (c6)  [right=of c5]  {$$};
\node[state]  (c7)  [right=of c6]  {$$};
\node[state]  (c8)  [below=of c4]  {$$};

\node[state]  (co1) [right=of c8]  {$t$};
\node[state]  (co2) [above=of co1]  {$$};
\node[state]  (r3)  [right=of co2]  {$t'$};
\node[state]  (r4)  [below=of r3]  {$s$};
\node[state]  (r5)  [below=of r4]  {$t''$};
\node[state,accepting]  (r6)  [below=of r5]  {$$};

\path[->] 
(prestart) edge node {$$} (start)
(start) edge node[above] {$\texttt{[add]}$} (r1)
(r1) edge node[above] {$\openshort{x}$} (r2)
(r2) edge node[above] {$\ta$} (n1)
(n1) edge node[above] {$\tzero$} (s1)
(s1) edge node[above] {$\_$} (s2)
(s2) edge node[above] {$\ta$} (s3)
(s3) edge node[left] {$\eword$} (s7)
(n1) edge node[below] {$\ta$} (s4)
(s4) edge node[below] {$\_$} (s5)
(s5) edge node[below] {$\tzero$} (s6)
(s6) edge node[left] {$\eword$} (s7)
(s7) edge node[above] {$\eword$} (c1)
(c1) edge node[above] {$\tzero$} (c2)
(c2) edge node[above] {$\_$} (c3)
(c3) edge node[above] {$\ta$} (c4)
(c4) edge node[left] {$\eword$} (c8)
(s7) edge node[below] {$\ta$} (c5)
(c5) edge node[below] {$\_$} (c6)
(c6) edge node[below] {$\tzero$} (c7)
(c7) edge node[below] {$\eword$} (c8)
(c8) edge node[above] {$\ta$} (co1)
(co1) edge node[right] {$\closeshort{x}$} (co2)
(co2) edge node[above] {$\texttt{;}$} (r3)
(r3) edge node[right] {$\texttt{[add]}$} (r4)
(r4) edge node[right] {$x$} (r5)
(r5) edge node[right] {$\texttt{;}$} (r6)
(start) edge[loop above] node {$\Sigma'$} (start)
(n1) edge[loop above] node {$\ta$} (n1)
(s1) edge[loop above] node {$\tzero$} (s1)
(s3) edge[loop above] node {$\ta$} (s3)
(s4) edge[loop below] node {$\ta$} (s4)
(s6) edge[loop below] node {$\tzero$} (s6)
(c1) edge[loop above] node {$\ta$} (c1)
(c2) edge[loop above] node {$\tzero$} (c2)
(c4) edge[loop above] node {$\ta$} (c4)
(c5) edge[loop below] node {$\ta$} (c5)
(c6) edge[loop below] node {$\ta$} (c6)
(c7) edge[loop below] node {$\tzero$} (c7)
(co1) edge[loop below] node {$\ta$} (co1)
(r3) edge[loop above] node {$\Sigma'$} (r3)
(r6) edge[loop left] node {$\Sigma$} (r6);

\end{tikzpicture}
}
\caption{An $\MFA$ for the regex $r$ for checking addresses. The symbol $\ta$ represents the character group $[\texttt{a}\text{-}\texttt{z}]$ (i.\,e., an arc labelled with $\ta$ represents individual arcs for each of the symbols in $[\texttt{a}\text{-}\texttt{z}]$) and the symbol $\tzero$ represents the character group $[\texttt{0}\text{-}\texttt{9}]$.}
\label{addressExampleFigure}
\end{figure}

We discuss a more elaborate example of a memory-deterministic regex.\par
Assume that we are dealing with input strings that have two occurrences of a postal address (both occurring between the symbols ``\texttt{[add]}'' and ``$\texttt{;}$''), for which we want to check whether they are identical and whether they have the right form. In the following definition of a regex performing this task, we use \emph{character groups} $[\texttt{a}\text{-}\texttt{z}]$ to denote the expression $(\texttt{a} \altop \texttt{b} \altop \ldots \altop \texttt{z})$, ``$\_$'' denotes the space character and $\Sigma$ is the complete alphabet. The regular expression $r_{\textsf{add}}$ that checks the correct format of the address is given by $r_{\textsf{add}} = r_{\textsf{name}} \cdot r_{\textsf{street}} \cdot r_{\textsf{city}} \cdot r_{\textsf{country}}$, where $r_{\textsf{name}} = r_{\textsf{country}} = [\texttt{a}\text{-}\texttt{z}]^+$ are simple expressions, but $r_{\textsf{street}} = ([\texttt{0}\text{-}\texttt{9}]^+\_[\texttt{a}\text{-}\texttt{z}]^+ \altop [\texttt{a}\text{-}\texttt{z}]^+\_[\texttt{0}\text{-}\texttt{9}]^+)$ and $r_{\textsf{city}} = ([\texttt{a}\text{-}\texttt{z}]^*[\texttt{0}\text{-}\texttt{9}]^+\_[\texttt{a}\text{-}\texttt{z}]^+ \altop [\texttt{a}\text{-}\texttt{z}]^+\_[\texttt{a}\text{-}\texttt{z}]^*[\texttt{0}\text{-}\texttt{9}]^+)$ must cater for the different possible orders of street name and street number (and also for the different possible orders of zip code and city name). The regex that solves the task described above is then given by $r = (\Sigma')^* \: \texttt{[add]} \: x\{r_{\textsf{add}}\} \: \texttt{;} \: (\Sigma')^* \: \texttt{[add]} \: x \: \texttt{;} \: \Sigma^*$, where $\Sigma' = \Sigma\setminus\{\texttt{[add]}, \texttt{;}\}$. Consider the $\MFA$ $M$ shown in Fig.~\ref{addressExampleFigure} that is equivalent to the regex $r$. For the sake of convenience, $M$ differs from $\canonicalMFA(r)$, but it can be easily obtained from $\canonicalMFA(r)$ by contracting the obvious paths of $\eword$-transitions. \par
It can be verified that $r$ is not a deterministic regex (in the sense of~\cite{FreydenbergerSchmidJCSS} and as defined on page~\pageref{deterministicDef}); this already follows from the fact that $r_{\textsf{add}}$ is not a deterministic (classical) regular expression. Moreover, $M$ is not a deterministic $\MFA$. In the following, we show that $M$ is indeed memory-deterministic. First, we recall that $M$ is memory-deterministic if, for all $q_1, q_2 \in Q$, $q_1 \samereach q_2$ implies $q_1 \memSync q_2$, where $q_1 \samereach q_2$ means that there is a word $w$ and synchronised computations $\vec{c}$ and $\vec{c'}$ of $M$ on input $w$ with $|\vec{c}| = |\vec{c'}| = m$ and the states of $\vec{c}[m]$ and $\vec{c'}[m]$ are $q_1$ and $q_2$, respectively, and $q_1 \memSync q_2$ means that the following properties are satisfied:

\begin{itemize}
\item For every reduced $C_1, C_2 \subseteq \memInstAlphabet_k$, $x \in \kSigma$ and $p_1, p_2 \in Q$, $((q_1, C_1, x) \to_{\contracted} p_1) \wedge ((q_2, C_2, x) \to_{\contracted} p_2) \Rightarrow (C_1 = C_2)$.
\item For every $x \in [k]$, 
\begin{itemize}
\item $\deltaContr(q_1, x) \neq \emptyset \Rightarrow (\deltaContr(q_2, y) = \emptyset$ for every $y \in \Sigma_k \setminus \{x\})$.
\item $\deltaContr(q_2, x) \neq \emptyset \Rightarrow (\deltaContr(q_1, y) = \emptyset$ for every $y \in \Sigma_k \setminus \{x\})$.
\end{itemize}
\end{itemize}

We note that there is only one state with a memory recall transition, namely $s$ (see Fig.~\ref{addressExampleFigure}). Moreover, there is no other state $s'$ with $s \neq s'$ and $s \samereach s'$ (this is due to the fact that any word that leads to state $s$ must have a suffix ``$\texttt{;} u \texttt{[add]}$'' with $u \in (\Sigma')^*$ and therefore cannot lead to any other state different from $s$). \par
Consequently, the property of memory-determinism can only be violated by states $q_1, q_2 \in Q$ with $q_1 \samereach q_2$ and $x \in \Sigma$, such that $(q_1, C_1, x) \to_{\contracted} p_1$ and $(q_2, C_2, x) \to_{\contracted} p_2$ with $C_1 \neq C_2$. For $q_1 = q_2$, this is obviously never the case. If $\{q_1, q_2\} \cap \{r, t\} = \emptyset$ (where $r$ and $t$ are as shown in Fig.~\ref{addressExampleFigure}), then $C_1 = C_2 = \emptyset$. If $q_1 = r$, then $q_1 \samereach q_2$ and $\texttt{[add]} \notin \Sigma'$ implies $q_2 = r$. Finally, if $q_1 = t$, then $x \in \{\ta, \texttt{;}\}$ must hold and we observe that $(t, \emptyset, \ta) \to_{\contracted} t$ and $(t, \{\close{x}\}, \texttt{;}) \to_{\contracted} t'$ are the only contracted transitions for $t$. All other states $p$ with some contracted transition $(p, C, \ta) \to_{\contracted} p'$ satisfy $C = \emptyset$ and for all other states $p$ with some contracted transition $(p, C, \texttt{;}) \to_{\contracted} p'$ we note that $t \notsamereach p$ (note that this latter case only applies to $p = t''$).\par
Consequently, $M$ is memory-deterministic and therefore $r$ is a memory-deterministic regex.

\end{document}